\documentclass[letterpaper, 10 pt, journal, twoside]{IEEEtran}
\usepackage{graphicx}
\usepackage{epsfig}
\usepackage{times}
\usepackage{amsmath}
\usepackage{amsthm}
\usepackage{thmtools,thm-restate}
\usepackage{amssymb}
\usepackage[section]{placeins} 
\usepackage{placeins}
\usepackage{amsmath} 
\usepackage{multirow}
\usepackage{graphicx}
\usepackage[normalem]{ulem}
\usepackage{subcaption}
\usepackage{algorithm,tabularx}
\usepackage{algpseudocode}
\usepackage{amsfonts}
\usepackage{cases}
\usepackage{romannum}
\usepackage{textcomp}
\usepackage{xcolor}%
\usepackage{textcomp}
\providecommand{\U}[1]{\protect\rule{.1in}{.1in}}
\providecommand{\U}[1]{\protect\rule{.1in}{.1in}}

\IEEEoverridecommandlockouts

\newtheorem{theorem}{Theorem}

\newtheorem{lemma}{Lemma}
\newtheorem{proposition}{Proposition}
\newtheorem{remark}{Remark}
\newtheorem{definition}{Definition}

\useunder{\uline}{\ul}{}
\makeatletter
\newcommand{\multiline}[1]{  \begin{tabularx}{\dimexpr\linewidth-\ALG@thistlm}[t]{@{}X@{}}
#1
\end{tabularx}
}
\makeatother
\usepackage{cite}

\usepackage{picins}
\usepackage{centernot}

\usepackage{enumitem}
\setlist[itemize]{leftmargin=*}
\usepackage{makecell}
\usepackage[normalem]{ulem}
\usepackage{cuted}
\usepackage{hyperref}
\usepackage{stfloats}
\usepackage{lscape}

\newcommand{\tsup}[1]{\textsuperscript{#1}}
\newcommand{\T}{\top}
\newcommand{\I}{\mathbf{I}}
\newcommand{\0}{\mathbf{0}}
\newcommand{\diag}{\text{diag}}
\newcommand{\bmtx}[1]{\begin{bmatrix}#1\end{bmatrix}}

\newcommand{\R}{\mathbb{R}}
\newcommand{\N}{\mathbb{N}}

\newcommand{\X}{\mathcal{X}}
\newcommand{\mb}[1]{\mathbf{#1}}

\newcommand{\y}{\bm{y}}

\IEEEoverridecommandlockouts

\begin{document}

\title{\bf Mesh Stability Guaranteed Rigid Body Networks Using Control and Topology Co-Design}


\author{Zihao Song, Shirantha Welikala, Panos J. Antsaklis and Hai Lin\thanks{This work was supported by the National Science Foundation under  Grant CNS-1830335 and Grant IIS-2007949. Zihao Song, Panos J. Antsaklis, and Hai Lin are with the Department of Electrical Engineering, University of Notre Dame, Notre Dame, IN, USA. e-mail: \{{\tt zsong2, pantsakl, hlin1}\}{\tt @nd.edu}. Shirantha Welikala is with the Department of Electrical and Computer Engineering, Stevens Institute of Technology, Hoboken, NJ, USA. e-mail: {\tt swelikal@stevens.edu}.}}
\maketitle
\thispagestyle{empty}

\begin{abstract}
Merging and splitting are of great significance for rigid body networks 
in making such networks reconfigurable. 
The main challenges lie in simultaneously ensuring the compositionality of the distributed controllers and the mesh stability of the entire network.
To this end, we propose a decentralized control and topology co-design method for rigid body networks, which enables flexible joining and leaving of rigid bodies without the need to redesign the controllers for the entire network after such maneuvers.
We first provide a centralized linear matrix inequality (LMI)-based control and topology co-design optimization of the rigid body networks with a formal mesh stability guarantee.
Then, these centralized mesh stability constraints are made decentralized by a proposed alternative set of sufficient conditions. 
Using these decentralized mesh stability constraints and Sylvester's criterion-based decentralization techniques, the said centralized LMI problem is equivalently broken down into a set of smaller decentralized LMI problems that can be solved at each rigid body, enabling flexible merging/splitting of rigid bodies. 
Finally, the effectiveness of the proposed co-design method is illustrated based on a specifically developed simulator and a comparison study with respect to a state-of-the-art method.
\end{abstract}

\section{Introduction}\label{sec:intro}


Mechanical devices, such as cars, quadrotors, satellites, and surface/underwater vehicles, are generally modeled as rigid-body dynamics. They can formulate rigid body networks and work as a team by grouping together with sophisticated coordination.
Rigid body networks are widely applied in many scenarios, such as surveillance \cite{ullah2021dynamic}, area coverage \cite{tnunay2023distributed}, and supply transportation \cite{yao2020control,parekh2022review}. 
This type of network often follows certain geometric formations with safe separations between them. In this way, they can result in spatially and temporally collective behaviors and gain more capabilities than individuals.

Past decades have observed several mainstream approaches for rigid body networks control, involving linear (e.g., PID \cite{seshan2021geometric}, LQR \cite{madeiras2024trajectory} and $H_{\infty}$ \cite{gong2020novel}), nonlinear (e.g., feedback linearization \cite{khan2020nonlinear,chen2021feedback}, consensus-based methods \cite{mondal2024fixed}, backstepping \cite{xia2023output} and sliding mode control \cite{dong2021anti}) and intelligent methods (e.g., model predictive control (MPC) \cite{ding2021representation}, fuzzy logic \cite{sarkhel2020fuzzy} and neural network \cite{liu2021adaptive,xia2023output}). 
Despite these developments, little attention has been paid to merging and splitting control for rigid body networks. Merging and splitting are basic maneuvers for rigid body networks to ensure the reconfigurability of the networks and enhance obstacle avoidance and robustness against faults. 

The main challenges of the merging and splitting control of rigid body networks lie in (1) the scalability with respect to the network size, (2) the distributed and compositional requirements on the controllers, and (3) mesh stability guarantee.
In particular, the scalability requires that the performance of the designed controllers should be preserved as the network size grows. Furthermore, the controllers are not only required to be distributed but also expected to be compositional with the mesh stability guarantee due to the change of topologies.
By compositionality, we mean that the designed controllers do not need to be redesigned after the merging and splitting of agents. As a generalization of the traditional string stability notion to general networks, mesh stability captures the non-increase of the tracking errors along the networks and, thus, guarantees the safety of the agents.
Existing works on merging and splitting control of the rigid body networks mainly rely on consensus-based \cite{mondal2024fixed}, MPC \cite{novoth2020distributed}, or other optimization-based methods \cite{zhu2022merging} and planning-based algorithms \cite{wei2022mpc}. 
However, these methods are either not scalable 
or not compositional. Besides, the mesh stability of the network is usually not ensured without the safety maintenance during the merging and splitting of agents.

Apart from the requirements on the controllers, communication topology, which determines the connectivity between rigid bodies, is another important component for rigid body network control.
In merging and splitting scenarios, the interconnection between rigid bodies may vary correspondingly with the neighboring sets of agents.
However, most of the existing works on the merging and splitting control assume the topologies are fixed \cite{novoth2020distributed,zhu2022merging,wei2022mpc} or dynamically switched \cite{wu2021distributed,wang2023robust}. Moreover, in these works, after such maneuverings, the topology may need to be redesigned for the entire network as the mesh stability may not hold, leading to safety concerns.



Based on the above discussion, the distributed controllers and the communication topologies are of equivalent significance for rigid body networks' merging and splitting. Besides, for the safety and flexibility of the network, both the mesh stability of the closed-loop network and the compositionality of the controllers are required.
To achieve these goals, instead of designing the control and topology separately, we believe a more effective way is to simultaneously design both. This leads to the problem of control and topology co-design for rigid body networks. 
In particular, we first formulate a centralized LMI-based control and topology co-design optimization from the stability, dissipativity, and mesh stability analysis of the rigid body network, where the passivity properties are obtained via local control synthesis. Then, by applying a Sylvester criterion-inspired decentralization technique, the original centralized LMI is broken down into a set of smaller LMIs that can be solved decentrally at each agent, and the original mesh stability constraints are made decentralized through a set of sufficient alternative conditions. A unique advantage of our proposed decentralized solution is that the controllers of the entire network do not need to be redesigned when agents merge and split, which enables seamless maneuvering of the agents.


We have studied the control and topology co-design problem for general networked systems in \cite{WelikalaP32022}, followed by its application in merging and splitting of vehicular platoons with $L_2$ weak string stability \cite{welikala2023dissipativity} and disturbance string stability \cite{Zihao2024CDC} guarantees. In this paper, we extend the application domain of our co-design framework to underactuated nonlinear rigid body networks in 3D spaces with formal mesh stability guarantees.
The main differences and contributions of this paper can be summarized as follows:
\begin{enumerate}
    \item 
    We generalize the control and topology co-design framework in \cite{Zihao2024CDC} to underactuated nonlinear rigid body networks in 3D spaces with formal mesh stability guarantees;
    \item The compositionality of our proposed co-design framework enables flexible merging and splitting of rigid body networks without the need for redesigning after such maneuvers;
    \item To obtain the local dissipativity properties, a novel local control synthesis optimization is proposed without the need to manually select the optimization parameters as compared to our previous work \cite{welikala2025decentralized};
    \item The effectiveness of our proposed co-design method is illustrated via our specifically developed simulator;
    \item Through the comparison to a state-of-the-art consensus-based method, the performance of our proposed co-design method is highlighted.
\end{enumerate}

\vspace{-1mm}
This paper is organized as follows. Some necessary preliminaries and the problem formulation are presented in Section \ref{sec:background} and \ref{sec:problem_formulation}, respectively. Our main results are presented in Section \ref{sec:main_results} and are supported by simulation examples discussed in Section \ref{sec:simulation}. Concluding remarks are provided in Section \ref{sec:conclusion}.

\section{Preliminaries}\label{sec:background}
\paragraph*{\textbf{Notations}}
The sets of real, natural, and positive real numbers are denoted by $\mathbb{R}$, $\mathbb{N}$, and $\mathbb{R}_+$, respectively. $\mathbb{R}^{n\times m}$ denotes the real matrices' space with $n$ rows and $m$ columns. An $n$-dimensional real column vector is denoted by $\mathbb{R}^n$.
Define $\mathcal{I}_N:=\{1, 2,...,N\}$ ($N\in \mathbb{N}$) as the index set. 
A block matrix $A\in\mathbb{R}^{n\times m}$ is represented as $A:=[A_{ij}]_{i\in\mathcal{I}_n, j\in\mathcal{I}_m}$, where $A_{ij}$ is the $(i,j)$\tsup{th} block of $A$ (for indexing purposes, either subscripts or superscripts may be used, i.e., $A_{ij} \equiv A^{ij}$). 
$[A_{ij}]_{j\in \mathcal{I}_m}$ and $\diag([A_{ii}]_{i\in\mathcal{I}_n})$ represent a block row matrix and a block diagonal matrix, respectively. We define $\{A^i\}:=\{A_{ii}\}\cup\{A_{ij},j\in\mathcal{I}_{i-1}\}\cup\{A_{ji}:j\in\mathcal{I}_{i-1}\}$. $SO(3)$ represents the special orthogonal group, i.e., $SO(3):=\{R\in\mathbb{R}^{3\times 3}|R^{\top}R=I,\mbox{det}(R)=1\}$.
To represent the cross product between vectors, we introduce the hat map $\hat{(\cdot)}:\mathbb{R}^3\rightarrow\mathfrak{so}(3)$ (Lie algebra) such that $\hat{x}y=x\times y$ for all $x$, $y\in\mathbb{R}^3$. The inverse of the hat map is defined as
$(\cdot)^{\vee}:\mathfrak{so}(3)\rightarrow\mathbb{R}^3$, which is denoted as the vee map.
The zero and identity matrices are denoted by $\0$ and $\I$, respectively (dimensions will be obvious from the context). 
The sum of a matrix $A$ and its transpose is defined as $\mathcal{S}(A):= A+A^\T$. 
$\mb{1}_{\{\cdot\}}$ is the indicator function and $\mathsf{e}_{ij} := \mb{1}_{\{i=j\}}$. 
The Euclidean norm of a vector is given by $|x|_2 := |x| := \sqrt{x^Tx}$. For matrix $A\in\mathbb{R}^{n\times m}$, its 1-norm and spectral norm are denoted by $\|A\|_1$ and $\|A\|$, respectively.
The $\mathcal{L}_2$ and $\mathcal{L}_{\infty}$ norms of a time-dependent vector are given by $\|x(\cdot)\|=\sqrt{\int_{0}^{\infty}|x(\tau)|^2d\tau}$ and $\|x(\cdot)\|_{\infty} = \sup_{t\geq 0}\ |x(t)|$, respectively. 
We use $\mathcal{K}$, $\mathcal{K}_{\infty}$, and $\mathcal{KL}$ to denote different classes of comparison functions, see, e.g., \cite{sontag1995characterizations}. 
For functions of time $t$, we omit denoting 
their dependence on $t$ when it is unnecessary for ease of expression.

\subsubsection{\textbf{Dissipativity Theory}} 

Consider the dynamics of a general system as 
\begin{equation}\label{Eq:GeneralSystem}
    \begin{cases}
        \dot{x} = f(x,u),  \\
        y=h(x,u),
    \end{cases}
\end{equation}
where $x\in\mathbb{R}^n$, $u\in\mathbb{R}^q$ and $y\in\mathbb{R}^m$ are the system state, input, and output, respectively. 
The function $f:\mathbb{R}^{n}\times\mathbb{R}^{q}\rightarrow\mathbb{R}^{n}$ represents the system dynamic mapping, 
and the function $h:\mathbb{R}^{n}\times\mathbb{R}^{q}\rightarrow\mathbb{R}^{m}$ is the output mapping.
$f$ is assumed to be locally Lipschitz continuous around each equilibrium point $x^* \in \X$ with $f(x^*,u^*) = \0, \forall x^* \in \X \subset \R^n$ ($\X$ denotes a set of equilibrium states, $u^*$ is the input at this equilibrium $x^*$) and both $u^*$ and $y^*:= h(x^*,u^*)$ being implicit functions of $x^*$. 
To analyze the dissipativity of \eqref{Eq:GeneralSystem} without the explicit knowledge of its equilibrium points, the \emph{X-equilibrium-independent dissipativity} ($X$-EID) property \cite{WelikalaP52022} is introduced below.

\begin{definition}(\textit{$X$-EID \cite{WelikalaP52022}})\label{Def:X-EID}
The system \eqref{Eq:GeneralSystem} is $X$-EID under supply rate $s:\R^{q}\times\R^{m}\rightarrow \R$ if there exists a continuously differentiable storage function $V:\R^{n}\times \X \rightarrow \R$ satisfying:
$V(x,x^*)>0$ with $x \neq x^*$, 
$V(x^*,x^*)=0$, and 
$$\dot{V}(x,x^*) = \nabla_x V(x,x^*)f(x,u)\leq  s(u-u^*,y-y^*),$$
for all $(x,x^*,u)\in\R^{n}\times \X \times \R^{q}$, where the supply rate $s$ is of the quadratic form characterized by a symmetric coefficient matrix $X:= [X^{kl}]_{k,l\in\mathcal{I}_2}\in\mathbb{R}^{q+m}$, i.e., 
\begin{center}
    $s(u-u^*,y-y^*):=  
    \bmtx{u-u^*\\y-y^*}^\T 
    \bmtx{X^{11} & X^{12}\\X^{21} & X^{22}}
    \bmtx{u-u^*\\y-y^*}.$
\end{center}
\end{definition}

Note that $X$-EID is equivalent to the conventional ($Q,S,R$)- dissipativity \cite{Willems1972a}, and thus, it also involves IF-OFP($\nu$,$\rho$) (input feedforward-output feedback passivity, with $X:=\scriptsize\bmtx{-\rho\I & \frac{1}{2}\I \\ \frac{1}{2}\I & -\nu\I}$) and L2G($\gamma$) (finite-gain $L_2$-stability, with $X:=\scriptsize\bmtx{\gamma^2\I & \0 \\ \0 & -\I}$) which will be used in this paper.


\subsubsection{\textbf{Network Configuration}} 

Consider a networked system $\Sigma$ comprised of $N$ decoupled subsystems $\{\Sigma_i: i\in\mathcal{I}_N\}$, where the dynamics of each subsystem $\Sigma_i, i\in\mathcal{I}_N$ are 
\begin{equation}\label{Eq:SubsystemDynamics}
    \Sigma_i:\begin{cases}
        \dot{x}_i = f_i(x_i,u_i),\\
        y_i = h_i(x_i,u_i),
    \end{cases}
\end{equation}
where $x_i\in\mathbb{R}^{n_i}$, $u_i\in\mathbb{R}^{q_i}$ and $y_i\in\mathbb{R}^{m_i}$ are the subsystem's state, input and output, respectively. 
In analogous to \eqref{Eq:GeneralSystem}, each $\Sigma_i$ \eqref{Eq:SubsystemDynamics} ($i\in\mathcal{I}_N$) is assumed to have a set $\mathcal{X}_i \subset \R^{n_i}$, where for every $x_i^* \in \mathcal{X}_i$, there is a unique $u_i^* \in \R^{q_i}$ that satisfies $f_i(x_i^*,u_i^*)=0$ while both $u_i^*$ and $y_i^*:= h_i(x_i^*,u_i^*)$ being implicit functions of $x_i^*$.

Note that if the control input of the subsystem \eqref{Eq:SubsystemDynamics} is designed as $u_i:=u_i(\{x_j\}_{j\in\mathcal{N}_i\cup\{i\}},w_i)$, for all $i\in\mathcal{I}_N$, where $w_i$ denotes external disturbances, and the selection of $u_i$ is determined by specific application scenarios, then the subsystems $\Sigma_i$ ($i\in\mathcal{I}_N$) are interconnected. In this way, if for the networked system $\Sigma$, its subsystems $\Sigma_i$ \eqref{Eq:SubsystemDynamics} are interconnected via the static interconnection matrix $M$ (as shown in Fig. \ref{Fig:Interconnection2}) with the relationship:
\begin{equation}\label{Eq:NSC2Interconnection}
    \bmtx{u\\z}= \bmtx{M_{uy} & M_{uw} \\ M_{zy} & M_{zw}}\bmtx{y\\w} \equiv  M \bmtx{y\\w},
\end{equation}
then the closed-loop networked system can be expressed as (also illustrated as in Fig. \ref{Fig:Interconnection2})
\begin{equation}\label{Eq:GeneralSystem_closedLoop}
   \Sigma:\begin{cases}
        \dot{x}=\mathcal{F}(x,w),  \\
        z = \mathcal{H}(x,w),
    \end{cases}
\end{equation}
where $x:=[x_i^\T]_{i\in\mathcal{I}_N}^\T\in\mathbb{R}^n$, $u:=[u_i^\T]_{i\in\mathcal{I}_N}^\T\in\mathbb{R}^q$, $z:=[z_i^\T]_{i\in\mathcal{I}_N}^\T\in\mathbb{R}^l$, $y:=[y_i^\T]_{i\in\mathcal{I}_N}^\T\in\mathbb{R}^m$ and $w:=[w_i^\T]_{i\in\mathcal{I}_N}^\T\in\mathbb{R}^r$ are the stacked system states, stacked control input, stacked performance output, stacked feedback output and stacked external disturbances with $q=\sum_{i\in\mathcal{I}_N} q_i$, $l=\sum_{i\in\mathcal{I}_N} l_i$, $m=\sum_{i\in\mathcal{I}_N} m_i$ and $r=\sum_{i\in\mathcal{I}_N} r_i$, respectively. 
While the blocks $M_{uw}$, $M_{zy}$ and $M_{zw}$ represent the mapping from disturbances to the input, output to performance output, and disturbances to performance output, respectively, the block $M_{uy}$ describes the interconnections between subsystems.
Here, $\mathcal{F}:\mathbb{R}^{n}\times\mathbb{R}^{r}\rightarrow\mathbb{R}^{n}$ is the closed-loop system dynamic mapping, $\mathcal{H}:\mathbb{R}^{n}\times\mathbb{R}^{r}\rightarrow\mathbb{R}^{l}$ is the stacked performance output mapping with $n:=\sum_{i\in\mathcal{I}_N} n_i$, and we assume $g(x_i^*,\0) = \0,\ \forall x^* \in \X$. 

In this way, the closed-loop networked system $\Sigma$ in \eqref{Eq:GeneralSystem_closedLoop} is similar to the form of \eqref{Eq:GeneralSystem}, and hence, the $X$-EID concept is applicable for the closed-loop networked system $\Sigma$ in \eqref{Eq:GeneralSystem_closedLoop}.
Moreover, we assume that each subsystem $\{\Sigma_i: i\in\mathcal{I}_N\}$ of the closed-loop networked system $\Sigma$ in \eqref{Eq:GeneralSystem_closedLoop} is $X_i$-EID, where $X_i = X_i^\T := [X_i^{kl}]_{k,l\in\mathcal{I}_2}$ (see Def. \ref{Def:X-EID}). 

%




\subsubsection{\textbf{X-EID-Based Topology Synthesis and Decentralization}}




\begin{figure}[!t]
    \centering
    \begin{subfigure}{0.21\textwidth}
        \includegraphics[width=\linewidth]{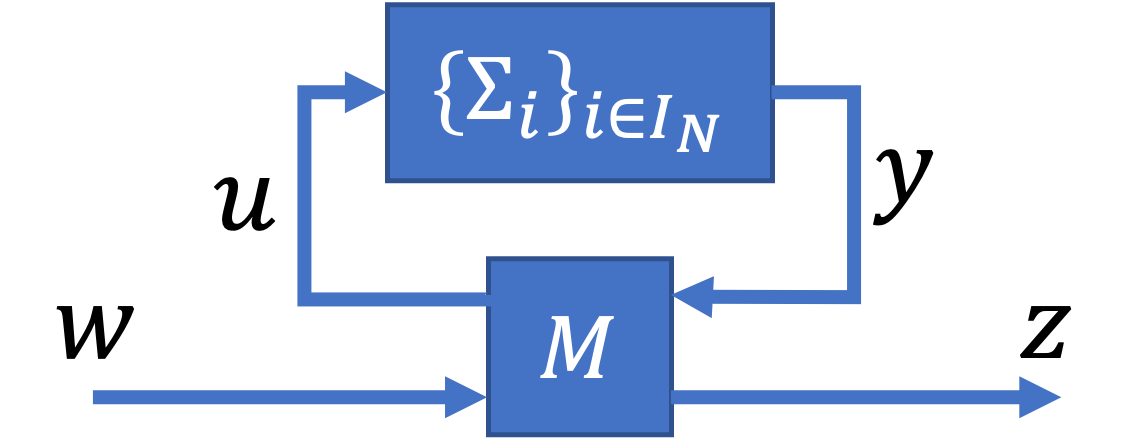}
        \vspace{-2mm}
        \caption{Networked system $\Sigma$}
        \label{Fig:Interconnection2}
    \end{subfigure}
    \begin{subfigure}{0.27\textwidth}
        \includegraphics[width=\linewidth]{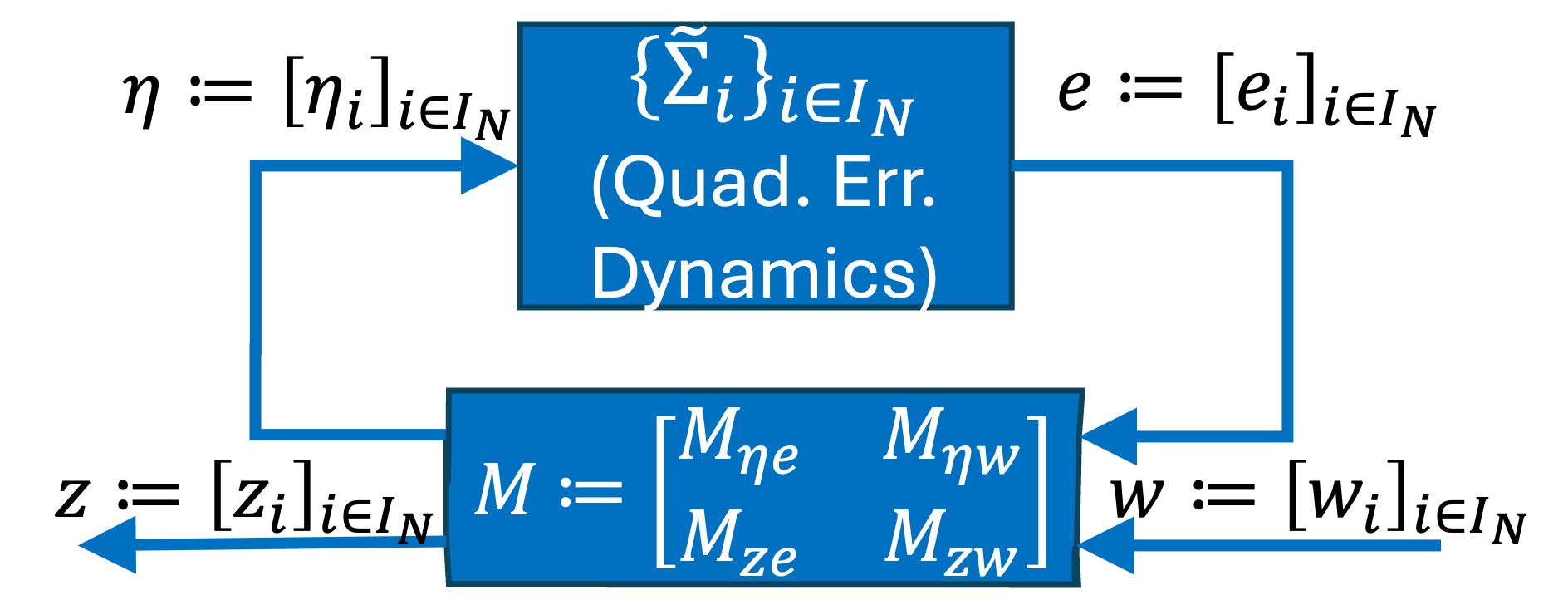}
        \vspace{-3mm}
        \caption{Formation error dynamics $\tilde{\Sigma}$}
        \label{Fig:FormationNetworkForm}
    \end{subfigure}
    \vspace{-6.5mm}
    \caption{
     Network configuration: (a) A generic networked system $\Sigma$; (b) Formation error dynamics as a networked system $\tilde{\Sigma}$.
     }
    \label{fig:networked_system_configuration}
\end{figure}

Based on the networked setup in \eqref{Eq:GeneralSystem_closedLoop} (also shown in Fig. \ref{Fig:Interconnection2}), we solve an LMI-based topology synthesis problem using the subsystem's $X_i$-EID properties to enforce specifications like $L_2$-stability as shown below.

\begin{proposition}\textit{\cite{WelikalaP52022}}\label{Pr:NSC2Synthesis}
The networked system $\Sigma$ in \eqref{Eq:GeneralSystem_closedLoop} can be made $L_2$-stable with the finite $L_2$-gain $\gamma$ by solving the following LMI problem to get the interconnection matrix $M$ in \eqref{Eq:NSC2Interconnection}:
\begin{equation}\label{Eq:Pr:NSC2Synthesis}
    \begin{aligned}
    & \mbox{Find: } L_{uy}, L_{uw}, M_{zy}, M_{zw}, \{p_i: i\in\mathcal{I}_N\}  \\
    & \mbox{s.t. } p_i > 0,\ \forall i\in\mathcal{I}_N,\  \mbox{and }   \\
	& {\scriptsize \bmtx{
		\textbf{X}_p^{11} & \0 & L_{uy} & L_{uw} \\
		\0 & \I & M_{zy} &  M_{zw}\\ 
		L_{uy}^\T & M_{zy}^\T & -L_{uy}^\T \textbf{X}^{12}-\textbf{X}^{21}L_{uy}-\textbf{X}_p^{22} & -\textbf{X}^{21}L_{uw} \\
		L_{uw}^\T & M_{zw}^\T & -L_{uw}^\T \textbf{X}^{12} &  \gamma^2\I
	} \normalsize >0},      
    \end{aligned}
\end{equation}
where $\textbf{X}^{12}:= \diag([(X_i^{11})^{-1}X_i^{12}]_{i\in\mathcal{I}_N})$ (we assume $X_i^{11}>0$) and $\textbf{X}^{21}:= (\textbf{X}^{12})^\T$ with 
$M_{uy}:= (\textbf{X}_p^{11})^{-1} L_{uy}$ and $M_{uw}:=  (\textbf{X}_p^{11})^{-1} L_{uw}$.
\end{proposition}


To break the centralized LMI in \eqref{Eq:Pr:NSC2Synthesis} into smaller LMIs and solve it in an equivalently decentralized manner, we recall Sylvester's criterion \cite{Antsaklis2006} inspired decentralization technique from \cite{WelikalaP32022}, i.e., a compositional verification of the positive definiteness of a symmetric block matrix.


\begin{proposition} \cite{WelikalaP32022}\label{Pr:MainProposition}
A symmetric $N \times N$ block matrix $W = [W_{ij}]_{i,j\in\mathcal{I}_N} > 0$ iff
\begin{equation}\label{Eq:Pr:MainProposition1}
    \tilde{W}_{ii} := W_{ii} - \tilde{W}_i \mathcal{D}_i \tilde{W}_i^\T > 0, \ \ \ \ \forall i\in\mathcal{I}_N,
\end{equation}
where $\tilde{W}_i := [\tilde{W}_{ij}]_{j\in\mathcal{I}_{i-1}} := W_i(\mathcal{D}_i\mathcal{A}_i^\T)^{-1}$, 
$W_i :=  [W_{ij}]_{j\in\mathcal{I}_{i-1}}$,
$\mathcal{D}_i := \diag([\tilde{W}_{jj}^{-1}]_{j\in\mathcal{I}_{i-1}})$, and $\mathcal{A}_i$ is the block lower-triangular matrix created from $[\tilde{W}_{kl}]_{k,l\in\mathcal{I}_{i-1}}$.
\end{proposition}

\subsubsection{\textbf{Mesh Stability}}

Similar to string stability in vehicular platoons, to capture the perturbation (e.g., external disturbances and tracking errors) propagation over general networked systems, we recall the \emph{Mesh Stability} concepts. 
For the networked system \eqref{Eq:GeneralSystem_closedLoop}, our mesh stability analysis is based on the time-domain (as opposed to frequency-domain) notion named \emph{scalable Mesh Stability (sMS)} introduced in \cite{mirabilio2021scalable}.


\begin{definition}(\textit{sMS \cite{mirabilio2021scalable}})\label{def:sMS}
    The networked system \eqref{Eq:GeneralSystem_closedLoop} around the equilibrium point $x^* \in \X$ is \emph{scalable mesh stable} (sMS), if there exist functions $\beta\in\mathcal{KL}$, $\sigma\in\mathcal{K}_{\infty}$, and constants $c_x$, $c_{w}>0$, such that for any initial condition $x_i(0)$ and disturbance $w_i$, $i\in\mathcal{I}_N$ satisfying
    \begin{equation}\label{Eq:Def:DSS_initial_bounds}
        \sup_{i\in\mathcal{I}_N}|x_i(0)-x_i^*|<c_x,\ \mbox{ and } \sup_{i\in\mathcal{I}_N}\|w_i\|_{\infty}<c_w,
    \end{equation}
    respectively, the solution $x_i(t)$ of the subsystem $\Sigma_i$ ($i\in\mathcal{I}_N$) of the networked system $\Sigma$ in \eqref{Eq:GeneralSystem_closedLoop} exists for all $t\geq 0$ and satisfies
    \begin{equation}\label{Eq:Def:DSS}
        \sup_{i\in\mathcal{I}_N} |x_i-x_i^*|\leq\beta(\sup_{i\in\mathcal{I}_N}|x_i(0)-x_i^*|,t)+\sigma (\sup_{i\in\mathcal{I}_N}\|w_i\|_{\infty}), 
    \end{equation}
    for all $t\geq 0$ and any $N\in\N$.
\end{definition}
\begin{remark}\label{Rm:conditions_for_sMS}    
The sMS concept in Def. \ref{def:sMS} indicates that as the tracking errors $|x_i-x_i^*|$ ($\forall i\in\mathcal{I}_N$) propagate over the network, they are uniformly bounded regardless of the total number of subsystems $N$. 
\end{remark}

\begin{proposition}\cite{mirabilio2021scalable}\label{Pr:sufficient_condition_DSS}
Suppose that each subsystem $\Sigma_i$ in the networked system \eqref{Eq:GeneralSystem_closedLoop} is input-to-state stable (ISS) that satisfies
\begin{equation*}\label{Eq:Def:ISS}
    |x_i-x_i^*|\leq\beta_i(|x_i(0)-x_i^*|,t)+\sigma_{xi}(\max_{j\in\mathcal{E}_i} \|x_j\|_{\infty})+\sigma_{wi}(\|w_i\|_{\infty}),
\end{equation*}
where $\beta_i\in\mathcal{KL}$, and $\sigma_{xi}$, $\sigma_{wi}\in\mathcal{K}_{\infty}$, for all $t\geq 0$,
and the conditions \eqref{Eq:Def:DSS_initial_bounds} hold for all $i\in\mathcal{I}_N$. Then, the closed-loop networked system \eqref{Eq:GeneralSystem_closedLoop} is sMS if there exist scalars $\bar{\sigma}_{xi}\in(0,1)$ such that 
\begin{equation}\label{Eq:Pr:weak_coupling_condition}
    \sigma_{xi}(s)\leq \bar{\sigma}_{xi}s
\end{equation}
holds for all $s\in\mathbb{R}_{\geq 0}$ and $i\in\mathcal{I}_N$.
\end{proposition}


\section{Error Dynamics Modeling and Problem Formulation}\label{sec:problem_formulation}
\subsection{Error Dynamics Modeling}

Consider a 3D rigid body dynamics of the $i\tsup{th}$ agent
$\Sigma_i$ in the network $\Sigma$ as \cite{lee2010geometric}:
    \begin{equation}\label{Eq:rigidBody_dynamics_i_SE3}
       \Sigma_{i}:
        \begin{cases}
        \dot{x}_i(t) = v_i(t),  &  \\
        m_i\dot{v}_i(t) = -f_i(t)R_i(t)e_3+m_ige_3+d_{vi}(t),  & \\
        \dot{R}_i(t) = R_i(t)\hat{\Omega}_i(t),  & \\
        J_i\dot{\Omega}_i(t) = -\hat{\Omega}_i(t) J_i\Omega_i(t)+M_i(t)+d_{\Omega i}(t),   &
        \end{cases}
    \end{equation}
for each $i\in\mathcal{I}_N$ ($N$ is not fixed), where $x_i(t), v_i(t)\in\mathbb{R}^3$ are the position and translational velocity of the rigid body in inertial frame; $\Omega_i(t)\in\mathbb{R}^3$ is the angular velocity in body-fixed frame, respectively; $R_i(t)\in SO(3)$ represents the orientation of the rigid body with respect to the world frame; $m_i\in\mathbb{R}_+$ is the mass of the rigid body; $J_i\in\mathbb{R}_+^{3\times 3}$ is the inertia matrix around its center of mass; $g=9.81\mbox{m/s}^2$ is the gravitational acceleration constant; $d_{vi}(t)$, $d_{\Omega i}(t)\in\mathbb{R}^3$ are the bounded time-varying external disturbances;
The thrust force $f_i(t)\in\mathbb{R}$ and the torque $M_i(t)\in\mathbb{R}^3$ are the actual control inputs that we want to design.
Here, $e_3:=[0,0,1]^{\top}$ is the basis vector of the vertical direction.

\begin{remark}
    Note that the model \eqref{Eq:rigidBody_dynamics_i_SE3} can be easily generalized to other rigid body dynamics such as satellites and underwater vehicles by refining the terms of thrust $f_iR_ie_3$, gravity $m_ige_3$, and torque $M_i$ \cite{roza2012motion}. 
\end{remark}

Suppose each rigid body tracks some given trajectory $(x_{di}(t),v_{di}(t))$, for all $i\in\mathcal{I}_N$, where $x_{di}(t)$, $v_{di}(t)\in\mathbb{R}^3$ are the desired position and translational velocity that are pre-defined as some smooth functions of time. 
Unlike the desired position $x_{di}$ and desired translational velocity $v_{di}$, the desired orientation $R_{di}(t)\in SO(3)$ and desired angular velocity $\Omega_{di}(t)\in\mathbb{R}^3$ are determined by the nominal thrust force $f_{di}$ applied (the exact form of $f_{di}$ will be introduced later).
Based on these desired signals, we define the position, translational velocity, orientation, and angular velocity tracking errors, respectively, as:
\begin{subequations}\label{Eq:tracking_errors}
\begin{gather}
    e_{xi} :=x_i-x_{di}, \ \  e_{vi} :=v_i-v_{di},  \label{Eq:e_xi_and_e_vi}  \\
    e_{Ri} := \frac{1}{2}(R_{ei}-R_{ei}^\T)^{\vee},\ \ e_{\Omega i} :=\Omega_i-R_{ei}^\T\Omega_{di},     \label{Eq:e_Ri_and_e_Omegai}
\end{gather}
\end{subequations}
where $R_{ei}:=R_{di}^\T R_i\in SO(3)$. The computation of $R_{di}$ and $\Omega_{di}$ will be introduced later on.

It is worth mentioning that the derivative of the orientation tracking errors is \cite{lee2010geometric}:
\begin{align}
    \dot{e}_{Ri} &= \frac{1}{2}(R_{di}^\T R_i\hat{e}_{\Omega i}+\hat{e}_{\Omega i}R_i^\T R_{di})^{\vee}  \nonumber  \\
        &= C(R_{di}^\T R_i)e_{\Omega i}\equiv C_i e_{\Omega i},
\end{align}
where we denote $C(\cdot):=\frac{1}{2}(\mbox{tr}[(\cdot)^\T] \I-(\cdot)^\T)$ and $C_i:=C(R_{di}^\T R_i)$.

Note that the system \eqref{Eq:rigidBody_dynamics_i_SE3} is underactuated. To overcome this challenge in the design process, we use the so-called nominal thrust force concept and recall the following proposition.
\begin{proposition}(Nominal Thrust Force \textit{\cite{lee2010geometric}})\label{Pr:nominal_thrust_force_underactuation}
    With the desired orientation $R_{di}$, the thrust force term $-f_iR_ie_3$ in \eqref{Eq:rigidBody_dynamics_i_SE3} can be equivalently written as $-f_iR_ie_3 = f_{di}-X_i$ (with $f_{di}$, $X_i\in\mathbb{R}^3$), where $f_{di}$ is the nominal thrust force that is free to be designed and
\begin{align}\label{Eq:Xi}
    X_i=|f_{di}|((e_3^\T R_{ei}e_3)R_ie_3-R_{di}e_3)
\end{align}
is the nonlinear coupling term caused by the underactuation of the rigid body. Finally, based on the designed nominal thrust force $f_{di}$, the actual thrust force $f_i$ ($i\in\mathcal{I}_N$) can now be obtained by
\begin{equation}\label{Eq:thrust_force_by_nominal_thrust_force}
    f_i = -f_{di}^\T R_i e_3 = (|f_{di}|R_{di}e_3)^\T R_ie_3.
\end{equation}
\end{proposition}

\begin{remark}\label{Rm:R_di_Omega_di_computation}
    In \eqref{Eq:e_Ri_and_e_Omegai}, the desired orientation is computed by 
    \begin{equation*}
        R_{di}:=\bmtx{-\frac{(\hat{b}^2_{d3i})b_{d1i}}{|(\hat{b}^2_{d3i})b_{d1i}|}, & \hat{b}_{d3i} b_{d1i}, & b_{d3i}}
    \end{equation*}
    based on the desired direction of the first and the third body-fixed axis, i.e., $b_{d1i}$, $b_{d3i}\in S^2:=\{b\in\mathbb{R}^3: |b|=1\}$. Here, we select $b_{d1i}=\frac{v_{di}}{|v_{di}|}$ and the desired direction of the third body-fixed axis $b_{d3i}$ is computed as $b_{d3i}=-\frac{f_{di}}{|f_{di}|}$. Using this $R_{di}$, we can then approximate the desired angular velocity by $\Omega_{di}=\frac{1}{2\delta t}(R_{di}(t_1)^\T R_{di}(t_2)-R_{di}(t_2)^\T R_{di}(t_1))^{\vee}$, where $\delta t:=t_2-t_1$ is a small time interval between two time steps $t_1$ and $t_2$ with $t_1$, $t_2\geq 0$. 
\end{remark}

Define the tracking error vector $e_i := [e_{xi}^{\top},e_{vi}^{\top},e_{Ri}^{\top},e_{\Omega i}^{\top}]^{\top}$, for all $i\in\mathcal{I}_N$. Then, the tracking error dynamics of the $i\tsup{th}$ rigid body are:
\begin{align}\label{Eq:rigidbody_error_dynamics_original}
    \tilde{\Sigma}_i:\begin{bmatrix}
        \dot{e}_{xi}   \\
        \dot{e}_{vi}   \\
        \dot{e}_{Ri}   \\
        \dot{e}_{\Omega i}
    \end{bmatrix}=&\begin{bmatrix}
        \0 & \I & \0 & \0    \\
        \0 & \0 & \0 & \0    \\
        \0 & \0 & \0 & C_i   \\
        \0 & \0 & \0 & \0
    \end{bmatrix}\begin{bmatrix}
        e_{xi}   \\
        e_{vi}   \\
        e_{Ri}   \\
        e_{\Omega i}
    \end{bmatrix}+\begin{bmatrix}
        \0 & \0  \\
        \I & \0  \\
        \0 & \0  \\
        \0 & \I
    \end{bmatrix}*    \nonumber    \\
    &  \bigg(\begin{bmatrix}
        u_{1i} \\
        u_{2i}
    \end{bmatrix}+\begin{bmatrix}
        -\frac{1}{m_i}X_i+d'_{vi} \\
        d'_{\Omega i}
    \end{bmatrix}\bigg),
\end{align}
where the dynamics of $e_{xi}$ and $e_{vi}$ are named as outer-loop error dynamics and the dynamics of $e_{Ri}$ and $e_{\Omega i}$ are named as inner-loop error dynamics. The new control input components $u_{1i}$, $u_{2i}$ are defined as
\begin{subequations}
    \begin{align}
    u_{1i}&:=\frac{1}{m_i}f_{di}-\dot{v}_{di}+ge_3,  \label{Eq:u_1i}  \\
    u_{2i}&:=-J_i^{-1}\hat{\Omega}_iJ_i\Omega_i+J_i^{-1}M_i+\hat{\Omega}_iR_i^\T R_{di}\Omega_{di}-     \nonumber \\
        & R_i^\T R_{di}\dot{\Omega}_{di},      \label{Eq:u_2i}
\end{align}
\end{subequations}
and disturbance components now become $d'_{vi}:=\frac{1}{m_i}d_{vi}$ and $d'_{\Omega i} := J_i^{-1}d_{\Omega i}$.

\begin{remark}
    For the control input \eqref{Eq:u_1i}, we cannot cancel out the term $-\frac{1}{m_i}X_i$ in \eqref{Eq:rigidbody_error_dynamics_original}, since the nonlinear coupling term $X_i$ is caused by underactuation and it involves the nominal thrust force $f_{di}$ as seen in \eqref{Eq:Xi}.
\end{remark}

\subsection{Problem Formulation}

The architecture of the error dynamics of \eqref{Eq:rigidbody_error_dynamics_original} is shown in Fig. \ref{fig:rigidBodyNetwork_Control_Architecture}, where the position and translational velocity error dynamics are named as outer-loop dynamics and the orientation and angular velocity error dynamics are named as inner-loop dynamics, due to the different time scale they follow.
Note that the control objective of the inner-loop error dynamics is to stabilize the tracking errors $e_{Ri}$ and $e_{\Omega i}$, i.e., to ensure that $e_{Ri}$, $e_{\Omega i}\rightarrow 0$ as $t\rightarrow\infty$. Hence, to achieve the control and topology co-design for the rigid body networks, we have to first stabilize the inner-loop error dynamics by introducing the control component $u_{2i}$ as \cite{lee2010geometric}:
\begin{equation}\label{Eq:u_2i_designed}
    u_{2i}=-k_{Ri}e_{Ri}-k_{\Omega i}e_{\Omega i},
\end{equation}
where the control parameters $k_{Ri}$ and $k_{\Omega i}$ are positive constants for the inner-loop error dynamics, for all $i\in\mathcal{I}_N$.
Substituting the control component $u_{2i}$ into \eqref{Eq:u_2i_designed}, the resulting torque controller $M_i$ in \eqref{Eq:u_2i} can ensure the exponential stability of the inner-loop error dynamics as shown in \cite{lee2010geometric}.

\begin{remark}
    The selection of the inner-loop control parameters $k_{Ri}$ and $k_{\Omega i}$ will impact the outer-loop tracking performance, since the inner-loop tracking errors will be propagated to the outer-loop error dynamics as seen in \eqref{Eq:rigidbody_error_dynamics_original}.
    However, due to different time scales for the inner- and outer-loop error dynamics, the inner-loop control parameters $k_{Ri}$ and $k_{\Omega i}$ cannot be designed simultaneously with the outer-loop control parameters. To find the optimal control parameters for the inner-loop, optimal control parameters may be found using metaheuristic optimization techniques as proposed in \cite{boualem2023evaluation}.
    The selection of inner-loop control parameters according to the outer-loop control parameters is out of the scope of this paper, but may be found in \cite{bertrand2011hierarchical}.
\end{remark}

\begin{figure}[!t]
    \centering
    \includegraphics[width=0.95\linewidth]{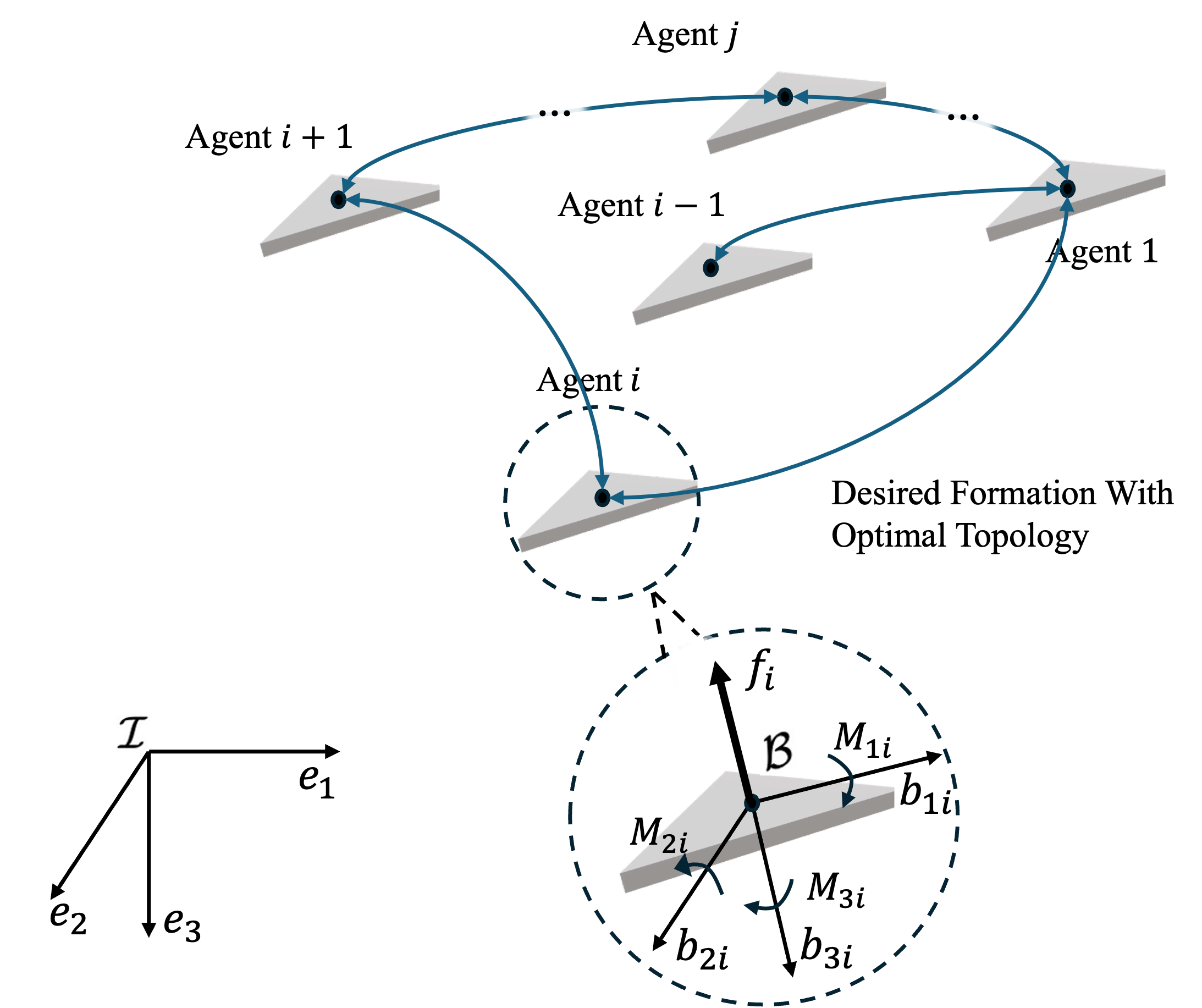}
    \caption{Configuration of the rigid body networks. Each agent is assumed to know the leader's information. }
    \label{fig:config_rigidbody_networks}
\end{figure}

\begin{figure}[!t]
    \centering
    \includegraphics[width=0.95\linewidth]{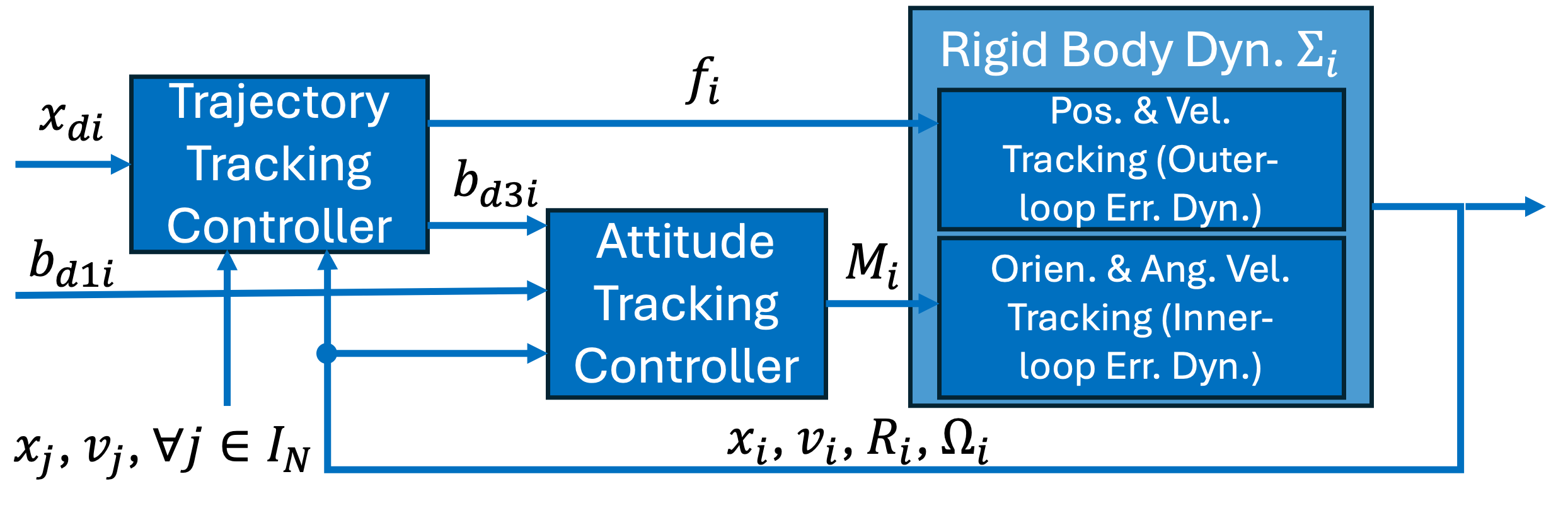}
    \caption{The error dynamics and control architecture for the $i\tsup{th}$ agent in the rigid body network. }
    \label{fig:rigidBodyNetwork_Control_Architecture}
\end{figure}

When the stability of the inner-loop error dynamics is guaranteed, we can focus on the control and topology co-design for the outer-loop error dynamics, i.e.,
\begin{align}\label{Eq:rigidbody_error_dynamics_outer-loop}
    \tilde{\Sigma}_{oi}:\begin{bmatrix}
        \dot{e}_{xi}   \\
        \dot{e}_{vi}   
    \end{bmatrix}=&\underbrace{\begin{bmatrix}
        \0 & \I   \\
        \0 & \0 
    \end{bmatrix}}_{A}\underbrace{\begin{bmatrix}
        e_{xi}   \\
        e_{vi}   
    \end{bmatrix}}_{\bar{e}_i}+\underbrace{\begin{bmatrix}
        \0  \\
        \I 
    \end{bmatrix}}_{B}u_{1i}+\underbrace{\begin{bmatrix}
        \0  \\
        -\frac{1}{m_i}X_i+d'_{vi} 
    \end{bmatrix}}_{w_i},     \nonumber  \\
    =&A\bar{e}_i+Bu_{1i}+w_i,
\end{align}
where the impact of the inner-loop error dynamics is reflected by the nonlinear coupling term $X_i$.

To achieve the tracking control goals while maintaining the sMS (see Def. \ref{def:sMS}) for the entire network, we design the control input component $u_{1i}$ ($i\in\mathcal{I}_N$) for the outer-loop error dynamics in \eqref{Eq:rigidbody_error_dynamics_original} as:
\begin{align}\label{Eq:FeedbackControlInput}
    u_{1i}&=(\bar{L}_{ii}+L_{ii}) \bar{e}_i(t) + \sum_{j\in\mathcal{I}_N\backslash\{i\}} L_{ij} (\bar{e}_i(t)-\bar{e}_j(t)), \nonumber \\
    &=\bar{L}_{ii}\bar{e}_i+\sum_{j\in\mathcal{I}_N}\bar{K}_{ij}\bar{e}_j,
\end{align}
where 
$\bar{L}_{ii}:=\bmtx{l_{ii}^{x} & l_{ii}^{v}}\in\mathbb{R}^{3\times 6}$ is used to regulate the local passivity properties; $\bar{K}_{ij}:=\bmtx{k_{ij}^{x} & k_{ij}^{v}}\in\mathbb{R}^{3\times 6}$ ($j\in\mathcal{I}_N$) are the distributed controller gains (topology) for the network.
It is required that $\bar{K}_{ij} := -L_{ij}, \forall j\neq i$ and $\bar{K}_{ii} := L_{ii} + \sum_{j\in\mathcal{I}_N\backslash\{i\}}L_{ij}$.

Due to different time scales for inner- and outer-loop error dynamics, we have to ensure that the outer-loop dynamics are not dramatically affected during the convergence of the inner-loop error dynamics.
Therefore, using the outer-loop control component \eqref{Eq:FeedbackControlInput}, we introduce the following lemma to ensure the boundedness of the nonlinear coupling term $X_i$ so that we can reasonably view this term as time-varying bounded disturbances in the outer-loop control and topology co-design process.
\begin{lemma}\label{Lem:Growth Restriction Condition}
    For the outer-loop error dynamics in \eqref{Eq:rigidbody_error_dynamics_outer-loop}, there exist a constant $c_{\Delta}\in\mathbb{R}_+$ and a function $\psi(\cdot)\in\mathcal{K}$, which is differentiable at $[e_{R}^\T,\ e_{\Omega }^\T]^\T=\0$ such that the coupling term $X$ presents the following condition:
    \vspace{-1mm}
    \begin{equation}\label{Eq:growth_restriction_condition}
        |X|\leq\psi(e_R)|\bar{e}|,\ \ \mbox{for}\ |\bar{e}|\geq c_{\Delta}, 
    \end{equation}
    where $\bar{e}:=[\bar{e}_i^\T]_{i\in\mathcal{I}_N}^\T$, $e_R:=[e_{Ri}^\T]_{i\in\mathcal{I}_N}^\T$ and $e_{\Omega}:=[e_{\Omega i}^\T]_{i\in\mathcal{I}_N}^\T$.
\end{lemma}
\begin{proof}
    If we stack up the nominal thrust forces $f_{di}$ with the designed components in \eqref{Eq:u_1i} and \eqref{Eq:FeedbackControlInput} for all $i\in\mathcal{I}_N$, i.e., $F_d:=[f_{di}^\T]_{i\in\mathcal{I}_N}^\T$, then we have:
    \begin{equation*}\label{Eq:norm_desired_Fd}
        \begin{split}
        |F_{d}|&=|M(\bar{L}\bar{e}+\bar{K}\bar{e}+\dot{v}_{d}-\vec{1}ge_3)|    \nonumber \\
        &\leq B+\sqrt{2}\max_i\{m_i\}\max\{|\lambda_{\bar{L}}|,|\lambda_{\bar{K}}|\}|\bar{e}|     \nonumber \\
        &= \frac{1}{2}k_f\left(c_f+|\bar{e}|\right),
        \end{split}
    \end{equation*}
where $M:=\mbox{diag}([m_i]_{i\in\mathcal{I}_N})$, $\bar{L}:=\mbox{diag}([\bar{L}_{ii}]_{i\in\mathcal{I}_N})$, $\bar{K}:=[\bar{K}_{ij}]_{i,j\in\mathcal{I}_N}$, $\vec{1}:=[\I,\I,...,\I]^\T\in\mathbb{R}^{3N\times 3}$, and we assume $|M\dot{v}_{d}-M\vec{1}ge_3|\leq B$ with some $B\in\mathbb{R}_+$ as \cite{lee2010geometric}. 
$\lambda_{\bar{L}}$ and $\lambda_{\bar{K}}$ are the maximum eigenvalues of the matrices $\bar{L}$ and $\bar{K}$, respectively.
$k_f=2\sqrt{2}\max_i\{m_i\}\max\{|\lambda_{\bar{L}}|,|\lambda_{\bar{K}}|\}$, $c_f=\frac{B}{\sqrt{2}\max_i\{m_i\}\max\{|\lambda_L|,|\lambda_{K}|\}}$. 
Thus, we have:
\begin{equation}\label{Eq:norm_two_Fd}
    |F_d|\leq\left\{\begin{array}{ll}
    k_f|\bar{e}|, & |\bar{e}|\geq c_f \\
    k_fc_f, & |\bar{e}|< c_f
    \end{array}\right..
\end{equation}


Hence, for the stacked coupling term $X:=[X_i^\T]_{i\in\mathcal{I}_N}^\T$, we have:
\begin{align*}
    & |X|=|[X_i^\T]_{i\in\mathcal{I}_N}^\T|=|[|X_i|]_{i\in\mathcal{I}_N}^\T|\leq |[|f_{di}||e_{Ri}|]_{i\in\mathcal{I}_N}^\T|   \nonumber  \\
    &\leq |\mbox{diag}([|e_{Ri}|]_{i\in\mathcal{I}_N}^\T)|_F |F_d| \ (\mbox{use matrices' Frobenius norm})\nonumber    \\
    &=|[|e_{Ri}|]_{i\in\mathcal{I}_N}^\T| |F_d|=|e_R||F_d|.
\end{align*}

Thus, based on \eqref{Eq:norm_two_Fd}, we can conclude that
\begin{equation}
    |X|\leq k_f|e_R||\bar{e}|,\ \mbox{for}\ |\bar{e}|\geq c_f,
\end{equation}
where the class-$\mathcal{K}$ function and the positive constant are $\psi(e_R)=k_f|e_R|$ and  $c_{\Delta}=c_f$ in \eqref{Eq:growth_restriction_condition}, respectively.
\end{proof}

With Lem. \ref{Lem:Growth Restriction Condition}, we can reasonably view the term $X_i$ as time-varying bounded disturbances in outer-loop error dynamics during the convergence of the inner-loop error dynamics as in \eqref{Eq:rigidbody_error_dynamics_outer-loop}.
Hence, substituting the controller \eqref{Eq:FeedbackControlInput} into the outer-loop error dynamics \eqref{Eq:rigidbody_error_dynamics_outer-loop}, we can rewrite the outer-loop tracking error dynamics of the rigid body network as:
\begin{equation}\label{Eq:closedLoopErrorDynamics}
    \tilde{\Sigma}_{oi}:\ \ \dot{\bar{e}}_i = (A + B\bar{L}_{ii})\bar{e}_i + \eta_i,
\end{equation}
where $\eta_i:=Bu_{1i}+w_i\equiv\sum_{j\in\mathcal{I}_N} K_{ij} \bar{e}_j+w_i$ with $K_{ij}:=\scriptsize\bmtx{\0 & \0 \\
k_{ij}^{x} & k_{ij}^{v}}$, for all $i\in\mathcal{I}_N$.

Thus, by defining $M_{\eta\bar{e}}:= [K_{ij}]_{i,j\in\mathcal{I}_N}$, $M_{\eta w}:= \I$, $M_{z\bar{e}}:= \I$, and $M_{zw}:= \0$, the closed-loop rigid body network \eqref{Eq:closedLoopErrorDynamics} takes the form of a networked system shown in Fig. \ref{Fig:FormationNetworkForm}.

From \eqref{Eq:closedLoopErrorDynamics}, synthesizing $M_{\eta e}$ will reveal both the desired distributed controllers and communication topology for the rigid body network.
In this paper, our objective is to propose a dissipativity-based control and topology co-design method for rigid body network, facilitating merging and splitting while ensuring the $l_2$-stability and the sMS for the entire network.

\section{Main Results}\label{sec:main_results}

In this section, we provide our main theoretical results. 
First, assuming local dissipativity properties of the rigid bodies as IF-OFP($\nu_i,\rho_i$), for all $i\in\mathcal{I}_N$, we formulate the centralized control and topology co-design problem as an LMI problem. Next, to ensure the assumed local dissipativity properties, we present a local control synthesis process. Eventually, we propose our decentralized co-design process, which enables seamless merging/splitting of rigid body networks.

\subsection{Centralized Control and Topology Co-design}

Based on Prop. \ref{Pr:NSC2Synthesis}, the centralized control and topology co-design problem can be formulated as an LMI problem as summarized in the following theorem. Notably, due to the systematic modelling approach applied to the considered rigid body network, this theorem parallels and generalizes \cite[Th. 1]{Zihao2024CDC}, which addresses control and Topology co-design in longitudinal vehicular platoons with string stability guarantees.


\begin{theorem}\label{Th:CentralizedTopologyDesign}
The closed-loop networked system \eqref{Eq:closedLoopErrorDynamics} under the control input $\eta_i$ can be made finite-gain $L_2$-stable with some $L_2$-gain $\gamma$ (where $\tilde{\gamma} := \gamma^2 < \bar{\gamma}$) from disturbance input $w$ to performance output $z$ and sMS, by synthesizing the interconnection matrix block $M_{\eta e}=[K_{ij}]_{i,j\in\mathcal{I}_N}$ (as in Fig. \ref{Fig:FormationNetworkForm}) via solving the centralized LMI problem:
\begin{subequations}\label{Eq:Th:CentralizedTopologyDesign}
\begin{equation}
\begin{aligned}
\min_{Q,\gamma,\{p_i: i\in\mathcal{I}_N\}}& \sum_{i,j\in\mathcal{I}_N} c_{ij} \Vert Q_{ij} \Vert_1 + c_0 \tilde{\gamma}, \\
\mbox{s.t. }\ \ \ & \hspace{-2mm} p_i > 0,\ \forall i\in\mathcal{I}_N, \  0 < \tilde{\gamma} < \bar{\gamma},
\end{aligned}
\end{equation}
\begin{gather}
    \scriptsize \bmtx{
\textbf{X}_p^{11} & \0 & Q & \textbf{X}_p^{11} \\
\0 & \I & \I & \0\\ 
Q^\T & \I & -Q^\T \textbf{X}^{12}-\textbf{X}^{21}Q-\textbf{X}_p^{22} & -\textbf{X}^{21}\textbf{X}_p^{11} \\
\textbf{X}_p^{11} & \0 & -\textbf{X}_p^{11} \textbf{X}^{12}&  \tilde{\gamma} \I
} \normalsize >0,   \label{Eq:Th:CentralizedTopologyDesignMain}    \\
\mathcal{S}(R_iQ_{ii})\leq p_i\nu_i\epsilon_i\I \ (R_i > 0),\ \ \forall i\in\mathcal{I}_N,   \label{Eq:Th:sMS_Con_K_ii}   \\
\sum_{j\in\mathcal{I}_N\backslash\{i\}}\|R_iQ_{ij}\|< 
    -p_i\nu_i\delta_i,\ \ \forall i\in\mathcal{I}_N, \label{Eq:Th:sMS_Con_K_ij}
\end{gather}
\end{subequations}
where $c_0>0$ is a pre-specified constant, $\delta_i:=\sqrt{\mu_i\lambda_{\min}(R_i)}$ with $\mu_i:=\frac{(\rho_i+\epsilon_i-1)}{\lambda_{\max}(R_i)}$, and $0<\delta_i<1$, $Q:=[Q_{ij}]_{i,j\in\mathcal{I}_N}$ shares the same structure as $M_{\eta e}$, $\textbf{X}^{12}:= \diag([-\frac{1}{2\nu_i}\I]_{i\in\mathcal{I}_N})$, $\textbf{X}^{21}:= (\textbf{X}^{12})^\T$,
$\textbf{X}_p^{11}:= \diag([-p_i\nu_i\I]_{i\in\mathcal{I}_N})$, 
$\textbf{X}_p^{22} := \diag([-p_i\rho_i\I]_{i\in\mathcal{I}_N})$, and 
$M_{\eta e} := (\textbf{X}_p^{11})^{-1} Q$. 
\end{theorem}
\begin{proof}
The proof of \eqref{Eq:Th:CentralizedTopologyDesignMain} follows directly from the interconnection topology synthesis in Prop. \ref{Pr:NSC2Synthesis} using the fact that each subsystem's dissipativity properties are IF-OFP($\nu_i$,$\rho_i$) and the networked system is desired to be $l_2$-stable with the gain $\gamma$.

For the proof of sMS, consider the network error dynamics \eqref{Eq:closedLoopErrorDynamics}, where we denote $\bar{A}_{ii}:= A + B \bar{L}_{ii}$, and the local controllers $\bar{L}_{ii}$ are assumed given together with a feasible $R_i>0$ so that $\mathcal{S}(\bar{A}_{ii}^\T R_i)\leq  -\rho_i \I$ holds.
Select a storage function $V_i:= \bar{e}_i^\T R_i \bar{e}_i$, and take the directional derivative along \eqref{Eq:closedLoopErrorDynamics} with the controller $\eta_i$ and the condition \eqref{Eq:Th:sMS_Con_K_ii}: 
\begin{align}
\dot{V}_i =&\ \bar{e}_i^\T (R_i\bar{A}_{ii} + \bar{A}^\T_{ii} R_i) \bar{e}_i+2\bar{e}_i^\T R_i \eta_i,  \label{Eq:Th:Vidot_Xi_intermediate}  \\
\leq&-\rho_i\bar{e}_i^\T \bar{e}_i+2\bar{e}_i^\T R_i\Big(\sum_{j\in\mathcal{I}_N}K_{ij}\bar{e}_{j} + w_i\Big) \nonumber \\  
=&\ -(\rho_i+\epsilon_i)\bar{e}_i^\T \bar{e}_i+            2\bar{e}_i^\T\Big(\sum_{j\in\mathcal{I}_N\backslash\{i\}}R_iK_{ij}\bar{e}_j+R_iw_i\Big). \nonumber
\end{align}

Then, based on the Cauchy–Schwarz inequality and completing the squares, we can respectively obtain
\begin{align}
    \dot{V}_i\leq&-(\rho_i+\epsilon_i)|\bar{e}_i|^2+2|\bar{e}_i|        \nonumber  \\
    &\Big|\sum_{j\in\mathcal{I}_N\backslash\{i\}}\|R_iK_{ij}\|\max_{j\in\mathcal{I}_N\backslash\{i\}}|\bar{e}_j|+\|R_i\||w_i|\Big|               \nonumber   \\
    \leq&-(\rho_i+\epsilon_i-1)|\bar{e}_i|^2+        \nonumber  \\
    &\Big|\sum_{j\in\mathcal{I}_N\backslash\{i\}}\|R_iK_{ij}\|\max_{j\in\mathcal{I}_N\backslash\{i\}}|\bar{e}_j|+\|R_i\||w_i|\Big|^2              \nonumber   \\
    \leq&-\mu_i V_i+  W_i,           \nonumber
\end{align}
where we use $\lambda_{\min}(R_i)|\bar{e}_i|^2\leq V_i\leq\lambda_{\max}(R_i)|\bar{e}_i|^2$, and we denote $\mu_i:=\frac{(\rho_i+\epsilon_i-1)}{\lambda_{\max}(R_i)}$, $W_i:= \big|\|R_i\||w_i|+\sum_{j\in\mathcal{I}_N\backslash\{i\}}\|R_iK_{ij}\|\max_{j\in\mathcal{I}_N\backslash\{i\}}|\bar{e}_j|\big|^2,$
with $\rho_i+\epsilon_i-1 > 0$.
This leads to 
\vspace{-1mm}
\begin{equation*}
\lambda_{\min}(R_i)|\bar{e}_i|^2\leq V_i\leq\frac{W_i}{\mu_i}+\Big(V_i(0)-\frac{W_i}{\mu_i}\Big)e^{-\mu_i t},
\vspace{-1mm}
\end{equation*}
which further implies that 
\begin{align}
|\bar{e}_i|\leq&\sqrt{\frac{1}{\lambda_{\min}(R_i)}\Big(\frac{W_i}{\mu_i}+\Big(V_i(0)-\frac{W_i}{\mu_i}\Big)e^{-\mu_i t}\Big)}             \nonumber  \\
\leq&\sqrt{\frac{1}{\lambda_{\min}(R_i)}\Big(\frac{W_i}{\mu_i}+\lambda_{\max}(R_i)|\bar{e}_i(0)|^2e^{-\mu_i t}\Big)}  \nonumber      \\
\leq&\sqrt{\frac{1}{\mu_i\lambda_{\min}(R_i)}}\Big(\sum_{j\in\mathcal{I}_N\backslash\{i\}}\|R_iK_{ij}\|\max_{j\in\mathcal{I}_N\backslash\{i\}}\|\bar{e}_j\|_\infty+     \nonumber   \\ 
&\|R_i\|\|w_i\|_{\infty}\Big)+\sqrt{\frac{\lambda_{\max}(R_i)}{\lambda_{\min}(R_i)}}|\bar{e}_i(0)|e^{-\frac{\mu_i}{2} t}, 
\label{Eq:Th:ISS_subsystem_i}
\end{align}
where we have respectively used the properties $(1-e^{-a}) < 1, \forall a>0$,  $\sqrt{(a^2+b^2)}\leq(a+b), \forall a,b>0$ and $|a| < \|a\|_{\infty}$. 

It is readily seen that \eqref{Eq:Th:ISS_subsystem_i} implies the ISS of the error dynamics $\tilde{\Sigma}_{oi}, i\in\mathcal{I}_N$ \eqref{Eq:closedLoopErrorDynamics}. Thus, based on the sufficient condition in Rmk. \ref{Rm:conditions_for_sMS}, the condition \eqref{Eq:Th:sMS_Con_K_ij} is required to guarantee the sMS of the network.
\end{proof}

\begin{remark}
    Here, we provide the direct relationship between the synthesized interconnection matrix block $[K_{ij}]_{i,j\in\mathcal{I}_N}$ in Thm. \ref{Th:CentralizedTopologyDesign} and the individual agent (global) controller gains required in $\eta_i$ of the error dynamics \eqref{Eq:closedLoopErrorDynamics}.
    In particular, the off-diagonal elements of $[K_{ij}]_{i,j\in\mathcal{I}_N}$ are $K_{ij} = \scriptsize\bmtx{\0 & \0 \\
    k_{ij}^{x} & k_{ij}^{v}}$, for all $i\in\mathcal{I}_N, j\in\mathcal{I}_N\backslash\{i\}$, while the diagonal elements are
\begin{equation}\label{Eq:ControllerGainsDiagonal}
    K_{ii} = K_{i0} - \sum_{j\in\mathcal{I}_N\backslash\{i\}} K_{ij},
\end{equation}
for all $i\in\mathcal{I}_N$, where each $ K_{i0} = \scriptsize\bmtx{\0 & \0 \\
k_{ii}^{x} & k_{ii}^{v}}$.
\end{remark}

\subsection{Local Control Synthesis}

Note that in \eqref{Eq:Th:CentralizedTopologyDesignMain}, local dissipativity properties $(\nu_i,\rho_i)$ and $L_2$-gain $\gamma_i$ are required to initiate this program. Therefore, we present a local control synthesis in the following theorem. 
Compared to our previous work \cite{welikala2025decentralized}, the main difference lies in removing the manual selection of the $p_i$ ($i\in\mathcal{I}_N$) values in the local control synthesis optimization as shown next.


\begin{theorem}\label{Th:LocalControllerDesign}
At each agent $\Sigma_i$, to ensure the IF-OFP($\nu_i,\rho_i$) of the closed-loop networked system $\tilde{\Sigma}_{oi}$ \eqref{Eq:closedLoopErrorDynamics},
the local controller gains $\bar{L}_{ii}$ in \eqref{Eq:FeedbackControlInput} are obtained via the LMI problem:
\begin{subequations}\label{Eq:Th:LocalControllerDesign_tilde}
    \begin{align}
    \mbox{Find: }&\ \tilde{L}_{ii},\ \tilde{P}_i,\ \tilde{\nu}_i,\ \rho_i,\ \tilde{p}_i,\ \tilde{\gamma}_i, \nonumber \\
    \mbox{s.t. }&\ \tilde{P}_i > 0,\ \rho_i>\underline{\rho}_i>0,\ \tilde{\nu}_i<\bar{\tilde{\nu}}_i<0,    \nonumber  \\
    & 
    \bmtx{\I & \tilde{P}_i & \0 \\
    \tilde{P}_i &-\mathcal{S}(A\tilde{P}_i + B \tilde{L}_{ii})& -\rho_i\I + \frac{1}{2}\tilde{P}_i\\
    \0 & -\rho_i\I + \frac{1}{2}\tilde{P}_i & -\tilde{\nu}_i\I} > 0,  \label{Eq:Th:LocalControllerDesign_LMI_condition_tilde}   \\
    &  \bmtx{
-\tilde{\nu}_i & 0 & 0 & -\tilde{\nu}_i \\
0 & \tilde{p}_i & \tilde{p}_i & 0\\ 
0 & \tilde{p}_i & 1 & -\frac{1}{2} \\
-\tilde{\nu}_i & 0 & -\frac{1}{2} &  \tilde{\gamma}_i
} \normalsize >0,   \label{Eq:Th:CentralizedTopologyDesignMain_tilde}
\end{align}
\end{subequations}
for all $i\in\mathcal{I}_N$, where $\nu_i:=\rho_i^{-1}\tilde{\nu}_i$, 
$p_i := (\rho_i\tilde{p}_i)^{-1}$, 
$\tilde{\gamma}:=(\rho_i^2\tilde{p}_i)^{-1}\tilde{\gamma}_i$,  
$P_i:=\rho_i^{-1}\tilde{P}_i$, 
$R_i:=P_i^{-1}$, and 
$\bar{L}_{ii}:=\tilde{L}_{ii}\tilde{P}_i^{-1}$.
\end{theorem}
\begin{proof}
    We first show \eqref{Eq:Th:LocalControllerDesign_LMI_condition_tilde} by noting that for the closed-loop error dynamics (subsystem) $\tilde{\Sigma}_{oi}$ as in \eqref{Eq:closedLoopErrorDynamics} being IF-OFP($\nu_i$,$\rho_i$), we have (as also seen in Thm. 2 in \cite{welikala2025decentralized}):
    \begin{equation}
        \bmtx{\rho_i^{-1}\I & P_i & \0 \\
    P_i &-\mathcal{S}(AP_i + B \bar{L}_{ii}P_i)& -\I + \frac{1}{2}P_i\\
    \0 & -\I + \frac{1}{2}P_i & -\nu_i\I} > 0, \label{Eq:Th:LocalControllerDesign_LMI_condition}
    \end{equation}
    for all $i\in\mathcal{I}_N$, where $P_i:=R_i^{-1}$.

    Hence, if we multiply $\rho_i$ on both sides of the LMI \eqref{Eq:Th:LocalControllerDesign_LMI_condition}, we have:
    \begin{equation*}
        \bmtx{\I & \rho_i P_i & \0 \\
    \rho_i P_i &-\mathcal{S}(A\rho_i P_i + B \rho_i\bar{L}_{ii}P_i)& -\rho_i\I + \frac{1}{2}\rho_i P_i\\
    \0 & -\rho_i\I + \frac{1}{2}\rho_i P_i & -\rho_i\nu_i\I} > 0, 
    \end{equation*}
    which is right the same as the condition \eqref{Eq:Th:LocalControllerDesign_LMI_condition_tilde}.
 
    Then, we show the condition \eqref{Eq:Th:CentralizedTopologyDesignMain_tilde}.
    Denote the matrix in the LMI \eqref{Eq:Th:CentralizedTopologyDesignMain} as $\Phi$ (i.e., $\Phi>0$), and we equivalently have:
    \begin{equation}
        \Phi>0\ \Leftrightarrow\ \bar{\Phi}:=[[\Phi_{k,l}^{ij}]_{k,l\in\mathcal{I}_4}]_{i,j\in\mathcal{I}_N}>0,
    \end{equation}
    where $\bar{\Phi}$ is the ``block element-wise" permutation of $\Phi$ (as seen in Lem. 6 in \cite{welikala2024decentralized}).
    
    Using the diagonal blocks of $\bar{\Phi}$, we can identify a set of necessary conditions for the main LMI condition \eqref{Eq:Th:CentralizedTopologyDesignMain} in Thm. \ref{Th:CentralizedTopologyDesign} (i.e., $\Phi>0$) to hold as
    \begin{equation*}
        \Phi>0\ \Leftrightarrow\ \bar{\Phi}>0\ \Rightarrow\ \{\bar{\Phi}_{ii}>0,\ \forall i\in\mathcal{I}_N\},
    \end{equation*}
    where each $\bar{\Phi}_{ii}:=[\Phi_{kl}^{ii}]_{k,l\in\mathcal{I}_4}$ takes the form as
    \begin{equation}\label{Eq:Phi_ii_bar}
        \bar{\Phi}_{ii}:=\bmtx{
        -p_i\nu_i & 0 & 0 & -p_i\nu_i \\
        0 & 1 & 1 & 0\\
        0 & 1 & p_i\rho_i & -\frac{1}{2}p_i \\
        -p_i\nu_i & 0 & -\frac{1}{2}p_i &  \tilde{\gamma}
        } \normalsize >0,
    \end{equation}
    since the $Q_{ij}$ ($i,j\in\mathcal{I}_N$) blocks in the $Q$ matrix of \eqref{Eq:Th:CentralizedTopologyDesignMain} have the same form as $K_{ij}$, but with $Q_{ij}:=\scriptsize\bmtx{\0 & \0 \\
    q_{ij}^{x} & q_{ij}^{v}}$ and $Q_{ij}=-p_i\nu_i K_{ij}$, and each $Q_{ii}$ block has zeros on its diagonal.

    Note that for \eqref{Eq:Phi_ii_bar}, we have the following equivalence:
    \begin{align}
        \bar{\Phi}_{ii}>0\ & \Leftrightarrow\ \tilde{\Phi}_{ii}:=\frac{1}{p_i\rho_i}\bar{\Phi}_{ii}>0, \nonumber  \\
       \Leftrightarrow\ D_i^\T\tilde{\Phi}_{ii}D_i&=\bmtx{-\rho_i\nu_i & 0 & 0 & -\rho_i\nu_i \\
       0 & \rho_i^{-1}p_i^{-1} & \rho_i^{-1}p_i^{-1} & 0 \\
       0 & \rho_i^{-1}p_i^{-1} & 1 & -\frac{1}{2} \\
       -\rho_i\nu_i & 0 & -\frac{1}{2} & \rho_i\tilde{\gamma}},  \nonumber   \\
       &=\bmtx{-\tilde{\nu}_i & 0 & 0 & -\tilde{\nu}_i \\
       0 & \tilde{p}_i & \tilde{p}_i & 0 \\
       0 & \tilde{p}_i & 1 & -\frac{1}{2} \\
       -\tilde{\nu}_i & 0 & -\frac{1}{2} & \tilde{\gamma}_i}>0,    \label{Eq:DT_Phi_ii_bar_D}
    \end{align}
    where $D_i:=\mbox{diag}([\rho_i,\ 1,\ 1,\ \rho_i])$. Note that \eqref{Eq:DT_Phi_ii_bar_D} is exactly the condition \eqref{Eq:Th:CentralizedTopologyDesignMain_tilde}, and thus, this completes the proof.
\end{proof}

\begin{remark}
    The main steps for the implementation of local controller design and centralized global co-design are: 
    
\noindent
\textbf{Step\,1:} Select some scalar parameters: $p_i>0,\forall i\in\mathcal{I}_N$;\\
\noindent
\textbf{Step\,2:} Synthesize local controllers via \eqref{Eq:Th:LocalControllerDesign_tilde};\\
\noindent
\textbf{Step\,3:} If \eqref{Eq:Th:LocalControllerDesign_tilde} is infeasible, return to \textbf{Step\,1};\\
\noindent
\textbf{Step\,4:} Syntesize global co-design using Thm. \ref{Th:CentralizedTopologyDesign} (or Thm. \ref{Th:DecentralizedTopologyDesign} for decentralized co-design).

Different from our previous work \cite{welikala2025decentralized}, for the local control synthesis in \eqref{Eq:Th:LocalControllerDesign_tilde}, we remove the manual selection of the $p_i$ ($i\in\mathcal{I}_N$) values as observed in \eqref{Eq:Th:CentralizedTopologyDesignMain_tilde}.
Note that, a similar four-step process can be applied in a decentralized fashion if Step 4 (i.e., global co-design) can be made decentralized. This is introduced next. 
\end{remark}

\subsection{Decentralized Co-design for Merging and Splitting}

To enable merging and splitting, we require each agent can solve the control and topology co-design \eqref{Eq:Th:CentralizedTopologyDesign} in an equivalently decentralized manner. In this way, when agents merge or split, the controllers (also the topologies) of the remaining agents do not need to be redesigned.
Based on Prop. \ref{Pr:MainProposition}, we can break the LMI in \eqref{Eq:Th:CentralizedTopologyDesignMain} into smaller LMIs so that each agent only need to solve a corresponding one.

\begin{theorem}\label{Th:DecentralizedTopologyDesign}
The closed-loop error dynamics of the network $\tilde{\Sigma}_{oi}$ can be made finite-gain $L_2$-stable with some $L_2$-gain $\gamma$ (where  $\tilde{\gamma} := \gamma^2 < \bar{\gamma}$) from disturbance input $w$ to performance output $z$, if at each agent $\Sigma_i, i\in\mathcal{I}_N$: 
(1) the local controller gains $\bar{L}_{ii}$ are from \eqref{Eq:Th:LocalControllerDesign_LMI_condition_tilde}, (2) the interconnection/global controller gain blocks $\{K^i\}$ are designed using the local LMI problem:
\begin{subequations}\label{Eq:Th:DecentralizedTopologyDesign}
\begin{equation}
\begin{aligned}
 \min_{\{Q^i\}, \hat{\gamma}_i, p_i}& \hspace{-2mm}\sum_{j\in\mathcal{I}_{i-1}}c_{ij}\Vert Q_{ij} \Vert_1 + c_{ji} \Vert Q_{ji} \Vert_1 + c_{0i}\hat{\gamma}_i + c_{i}\vert \hat{\gamma}_i-\tilde{\gamma}_i\vert      \\ 
    \mbox{s.t. }\ \ \ & p_i>0,\ \hat{\gamma}_i < \bar{\gamma},\ \tilde{W}_{ii}>0,\  \eqref{Eq:Th:sMS_Con_K_ii},
\end{aligned}
\end{equation}
\begin{equation}\label{Eq:Th:sMS_Con_K_ij_Decentralized}
    \frac{1}{\delta_i}\|R_iQ_{ij}\|\leq -\frac{p_i\nu_i}{2^j},\ \ \forall j\in\mathcal{I}_{i-1}
\end{equation}
\end{subequations}
where $\tilde{\gamma}_i$ is from \eqref{Eq:Th:CentralizedTopologyDesignMain_tilde} (obtained in Step 2), and $\tilde{W}_{ii}$ is from \eqref{Eq:Pr:MainProposition1} when enforcing $W=[W_{ij}]_{i,j\in \mathcal{I}_N} > 0$ with 
\begin{equation*}
W_{ij} := 
\scriptsize \bmtx{
\mathsf{e}_{ij}V_p^{ii} & \0 & Q_{ij} & \mathsf{e}_{ij}V_p^{ii} \\
\0 & \mathsf{e}_{ij}\I & \mathsf{e}_{ij}\I & \0\\ 
Q_{ji}^\T & \mathsf{e}_{ij}\I & -Q_{ji}^\T S_{jj}-S_{ii}Q_{ij}-\mathsf{e}_{ij}R_p^{ii} & -\mathsf{e}_{ij}S_{ii}V_p^{ii} \\
\mathsf{e}_{ij}V_p^{ii} & \0 & -\mathsf{e}_{ij}V_p^{ii} S_{jj} &  \hat{\gamma}_i\mathsf{e}_{ij} \I
},
\end{equation*}
$V_p^{ii} := -p_i\nu_i \I$,
$R_p^{ii} := -p_i\rho_i \I$,
$S_{ii}:= -\frac{1}{2\nu_i}\I$ and blocks $\{K^{i}\}$ are determined by $K_{ij} = (V_p^{ii})^{-1}Q_{ij}$, and 
(2) the update: 
\begin{equation}\label{Eq:Th:DecentralizedTopologyDesign2}
 K_{j0}^{\mbox{New}} := K_{j0}^{\mbox{Old}} + K_{ji}   
\end{equation}
is requested at each prior and neighboring agent. 
\end{theorem}
\begin{proof}
    At each agent/iteration $\Sigma_i,i\in\mathcal{I}_N$, the set of blocks $\{K^i\}$ is derived. 
    By the network matrices-based decentralization technique in Prop. \ref{Pr:MainProposition}, enforcing $\tilde{W}_{ii}>0$ at each $\Sigma_i$ is equivalent to enforcing $W=[W_{ij}]_{i,j\in\mathcal{I}_N} > 0$ for the entire platoon $\Sigma$. 
    Note that, due to the special dependence \eqref{Eq:ControllerGainsDiagonal}, each derived $K_{ji}, j\in\mathcal{I}_{i-1}$, will affect the matrix $K_{jj}$ derived previously at the prior neighboring vehicle $\Sigma_j$ (of vehicle $\Sigma_i$) - violating the requirement that the $K$ matrix should be a network matrix. To ensure the network matrix property of $K$ (and thus, the application of Prop. \ref{Pr:MainProposition}), we need to use the following update:
    \begin{equation}
        K_{jj}^{\text{New}} = \Big(K_{j0} - \sum_{l<i, l \neq j} K_{jl}\Big) - K_{ji} = K_{jj}^{\text{Old}} - K_{ji},
    \end{equation}
    which requires the updates in \eqref{Eq:Th:DecentralizedTopologyDesign2}. This completes the proof of the decentralized LMI as in \eqref{Eq:Th:DecentralizedTopologyDesign}.

    For the proof of sMS, it only suffices to show that our proposed decentralized alternative \eqref{Eq:Th:sMS_Con_K_ij_Decentralized} implies the centralized version \eqref{Eq:Th:sMS_Con_K_ij}. From \eqref{Eq:Th:sMS_Con_K_ij_Decentralized} and the relation $K_{ij} = (V_p^{ii})^{-1}Q_{ij}$, it implies $\frac{1}{\delta_i}\sum_{j\in\mathcal{I}_N\backslash\{i\}}\|R_iK_{ij}\|\leq \sum_{j\in\mathcal{I}_N\backslash\{i\}} \frac{1}{2^j}$.

Basically, we index each ``$\frac{1}{2}$" with the same order of the $i\tsup{th}$ vehicle's neighbors. Since $\sum_{j\in\mathcal{I}_N\backslash\{i\}} \frac{1}{2^j}< \sum_{j\in\mathcal{I}_{\infty}} \frac{1}{2^j}=1$, \eqref{Eq:Th:sMS_Con_K_ij_Decentralized} implies \eqref{Eq:Th:sMS_Con_K_ij}. This completes the proof.
\end{proof}

\section{Simulation Examples}\label{sec:simulation}

\begin{figure}[!b]
\centering
\includegraphics[width=3in]{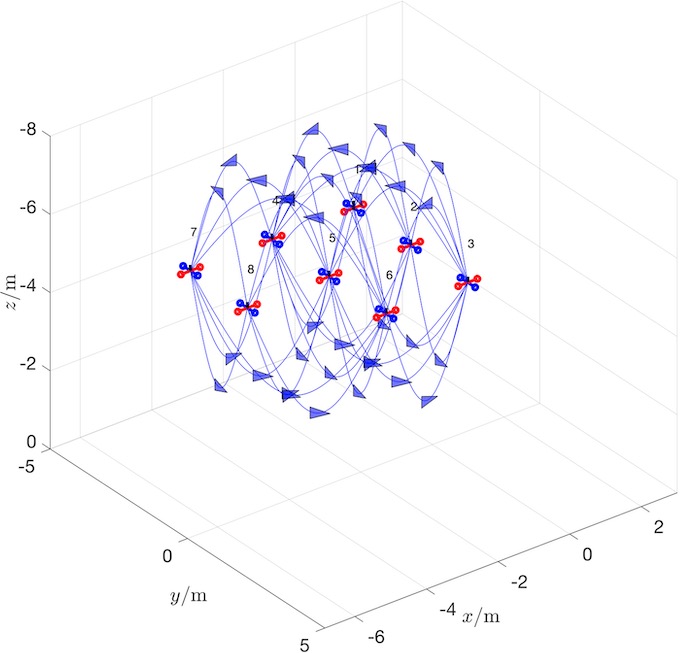}
\caption{Initially assumed communication topology.} 
\label{Fig:initialTopology_quadrotor_formation}
\end{figure}

\begin{figure*}[!t]
    \centering
    \begin{subfigure}{0.32\textwidth}
        \includegraphics[width=\linewidth]{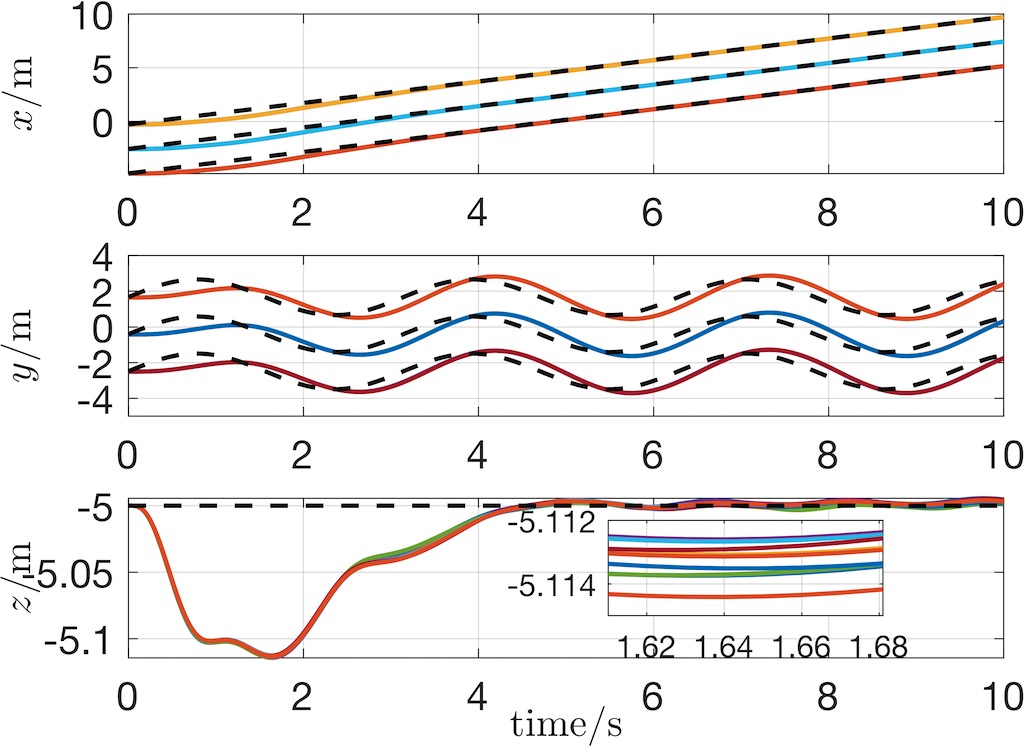}
        \vspace{-7mm}
        \caption{}
        \label{fig:position_tracking_results_quadrotor}
    \end{subfigure}
    \hfill
    \begin{subfigure}{0.32\textwidth}
        \includegraphics[width=\linewidth]{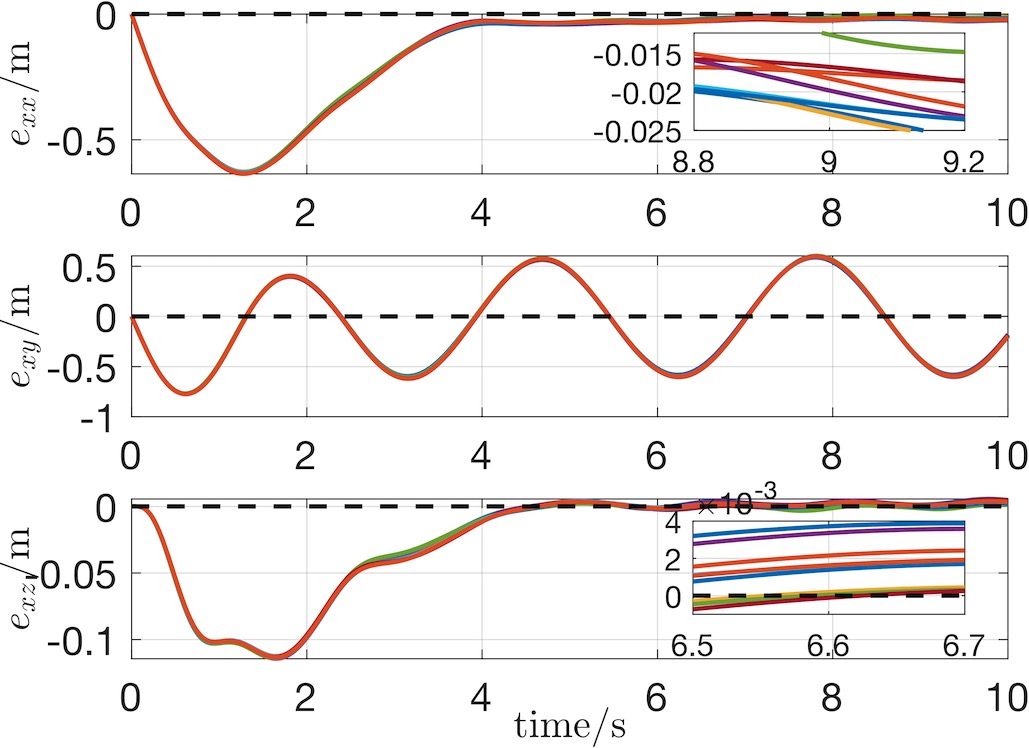}
        \vspace{-7mm}
        \caption{}
        \label{fig:position_tracking_errors_quadrotor}
    \end{subfigure}
    \hfill
    \begin{subfigure}{0.32\textwidth}
        \includegraphics[width=\linewidth]{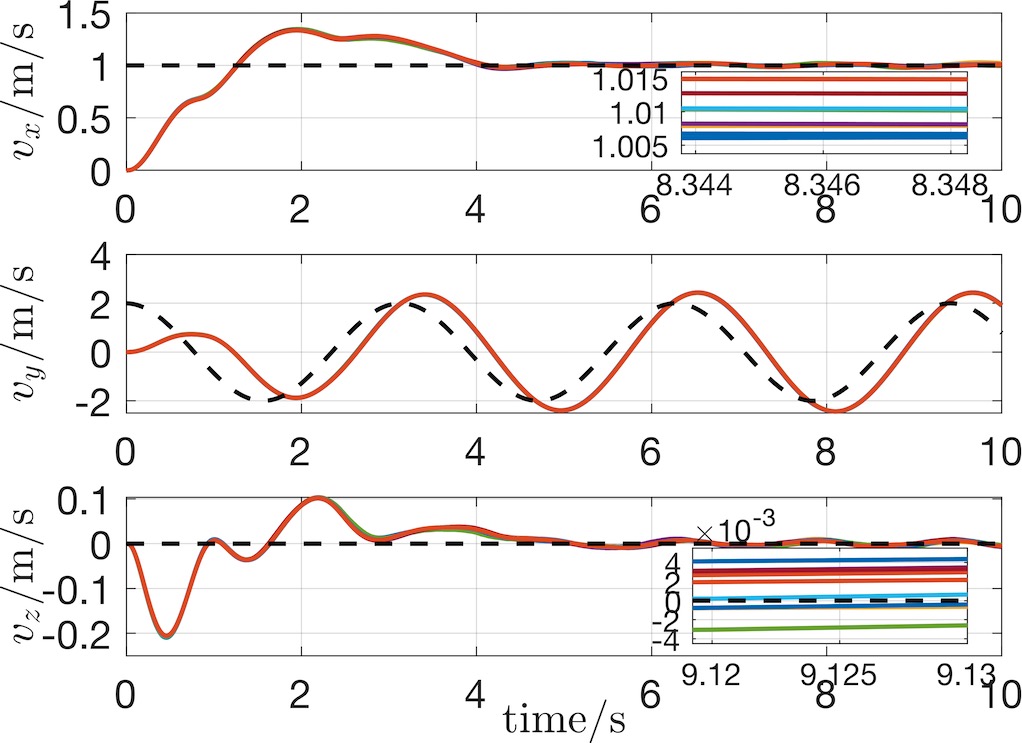}
        \vspace{-7mm}
        \caption{}
        \label{fig:velocity_tracking_results_quadrotor}
    \end{subfigure}
    \hfill
    \newline
    \begin{subfigure}{0.32\textwidth}
        \includegraphics[width=\linewidth]{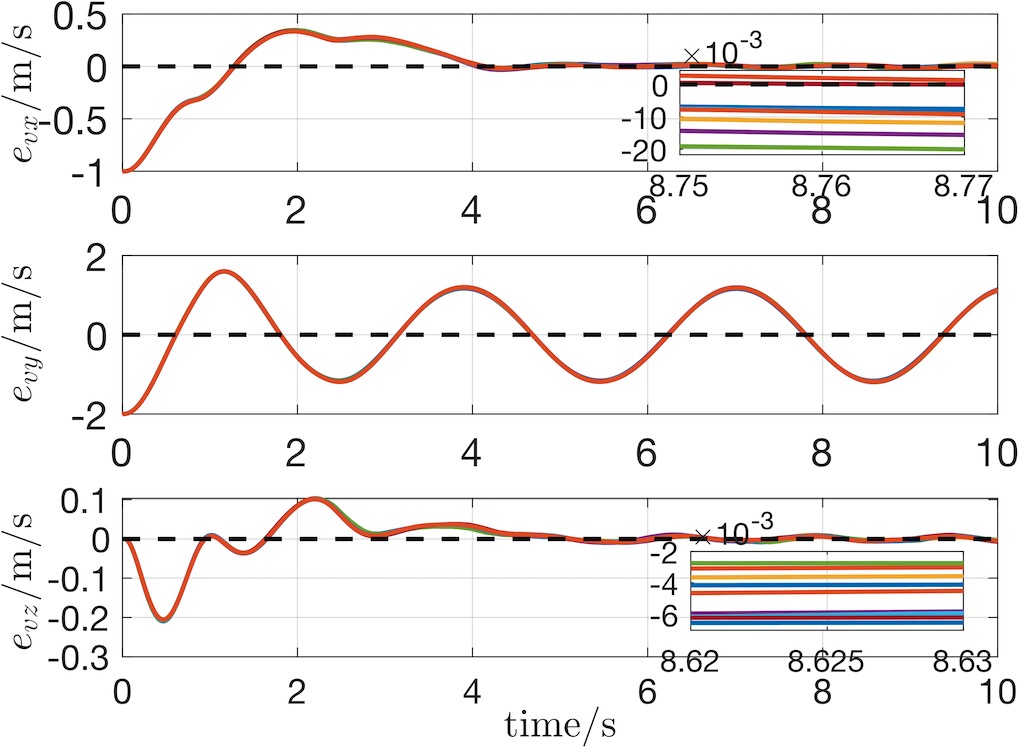}
        \vspace{-7mm}
        \caption{}
        \label{fig:velocity_tracking_errors_quadrotor}
    \end{subfigure}
    \hfill
    \begin{subfigure}{0.32\textwidth}
        \includegraphics[width=\linewidth]{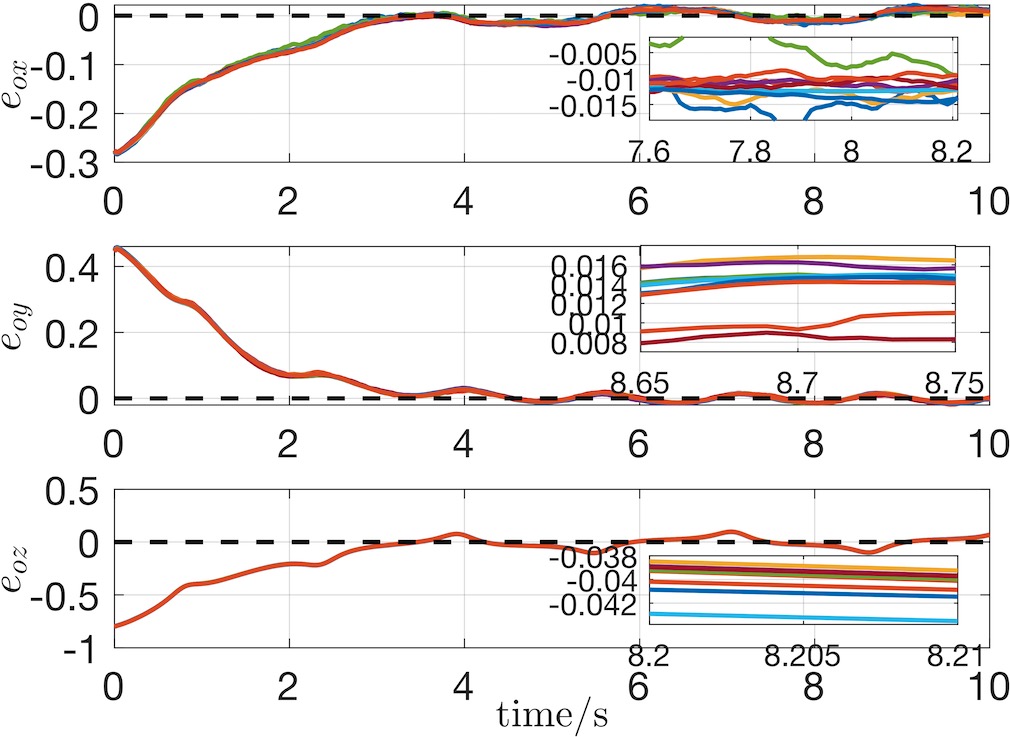}
        \vspace{-7mm}
        \caption{}
        \label{fig:orientation_tracking_errors_quadrotor}
    \end{subfigure}
    \hfill
    \begin{subfigure}{0.32\textwidth}
        \includegraphics[width=\linewidth]{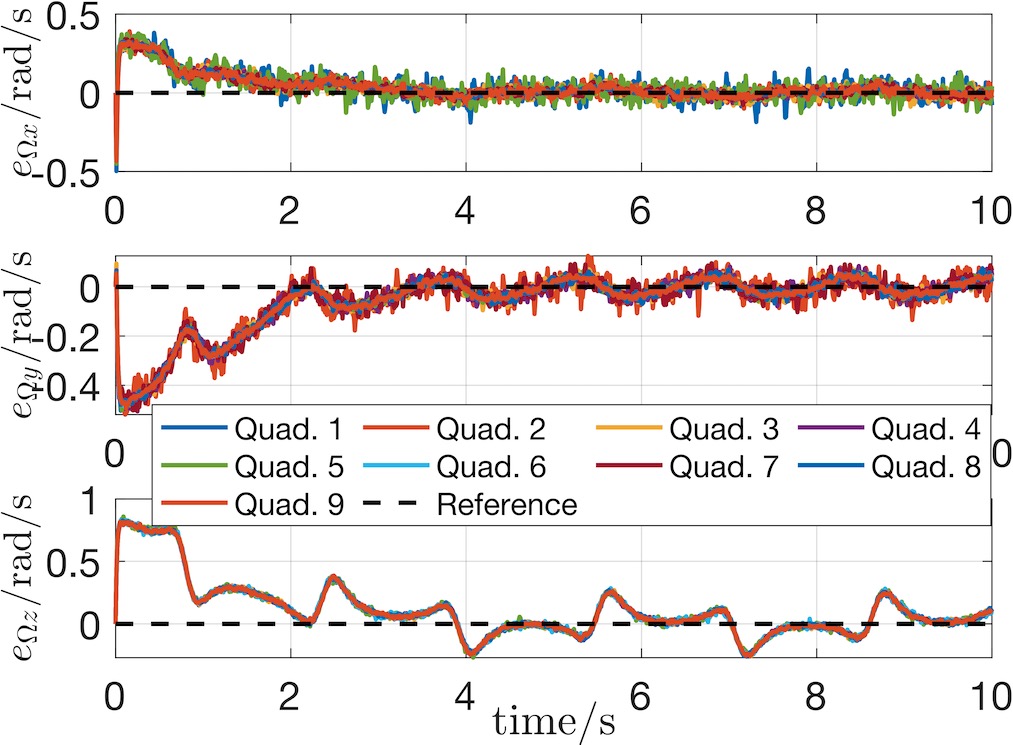}
        \vspace{-7mm}
        \caption{}
        \label{fig:angular_velocity_tracking_errors_quadrotor}
    \end{subfigure}
    \caption{Results observed by our proposed decentralized co-design with $9$ followers: (a) position tracking; (b) position tracking errors; (c) translational velocity tracking; (d) translational velocity tracking errors;
    (e) orientation tracking errors; (f) angular velocity tracking errors.}
    \label{fig:decentralized_merging_sMS_results}
\end{figure*}

\begin{figure*}[!t]
    \centering
    \begin{subfigure}{0.32\textwidth}
        \includegraphics[width=\linewidth]{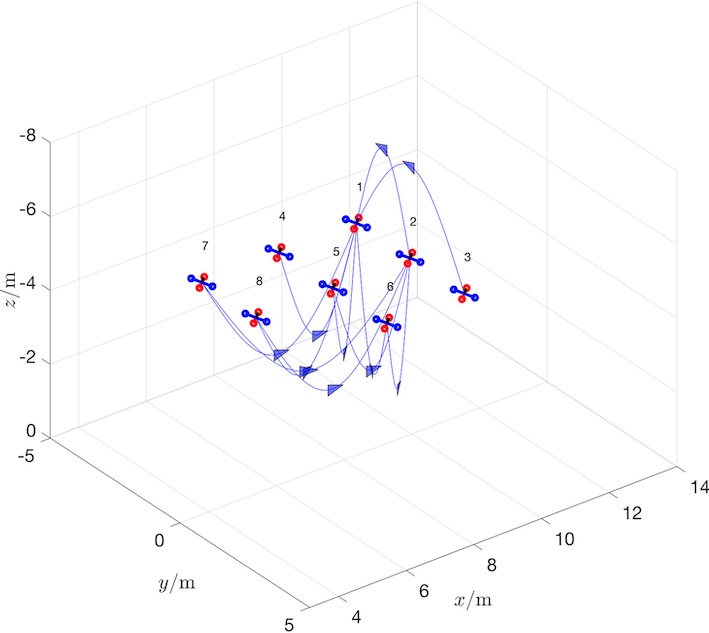}
        \vspace{-7mm}
        \caption{}
        \label{fig:decentralized_topology_sMS_8quadrotors}
    \end{subfigure}
    \hfill
    \begin{subfigure}{0.32\textwidth}
        \includegraphics[width=\linewidth]{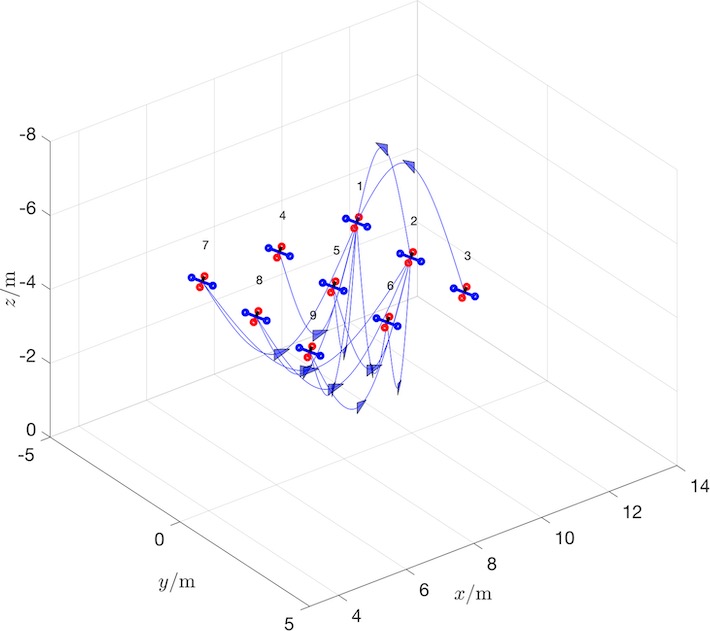}
        \vspace{-7mm}
        \caption{}
        \label{fig:decentralized_topology_sMS_9quadrotors}
    \end{subfigure}
    \hfill
    \begin{subfigure}{0.32\textwidth}
        \includegraphics[width=\linewidth]{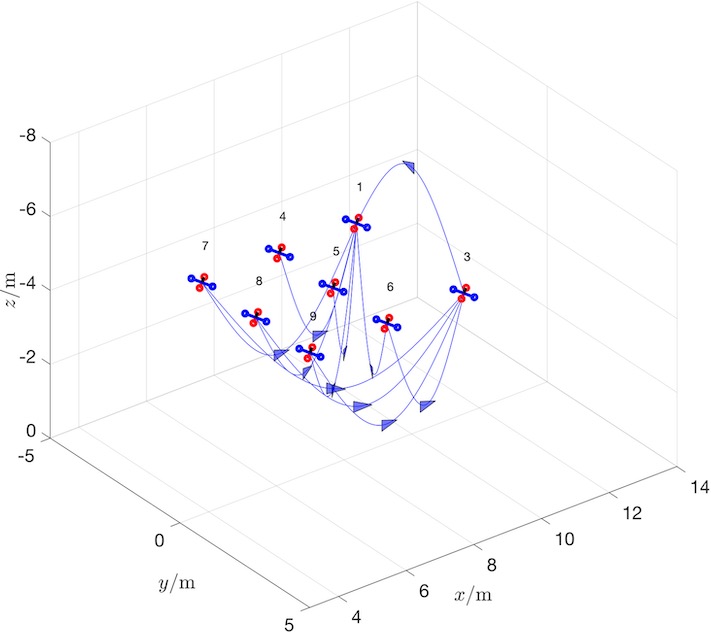}
        \vspace{-7mm}
        \caption{}
        \label{fig:decentralized_topology_sMS_8quadrotors_onemissing}
    \end{subfigure}
    \caption{
      Quadrotor merging and splitting process using our proposed decentralized co-design: (a) topology with $8$ followers; (b) add the $9$th quadrotor at the corner; (c) remove the $2$nd quadrotor from the $9$ quadrotors' formation.}
    \label{Fig:decentralized_sMS_quadrotors_topology}
\end{figure*}

\begin{figure*}[!t]
    \centering
    \begin{subfigure}{0.32\textwidth}
        \includegraphics[width=\linewidth]{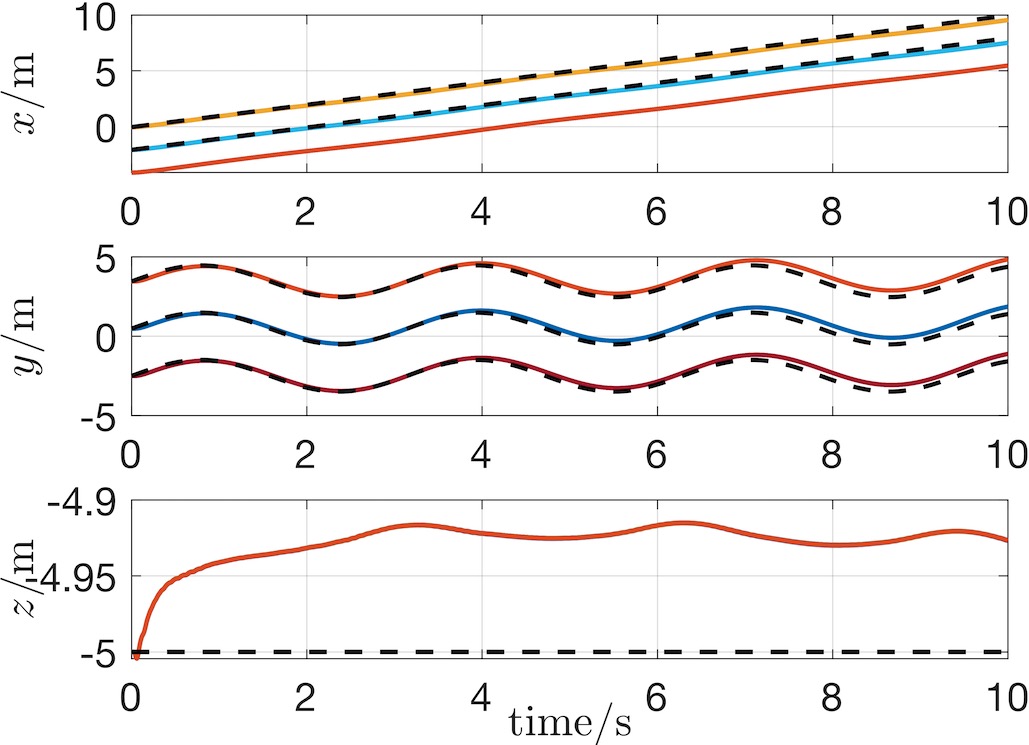}
        \vspace{-7mm}
        \caption{}
        \label{fig:position_tracking_results_quadrotor_consensus}
    \end{subfigure}
    \hfill
    \begin{subfigure}{0.32\textwidth}
        \includegraphics[width=\linewidth]{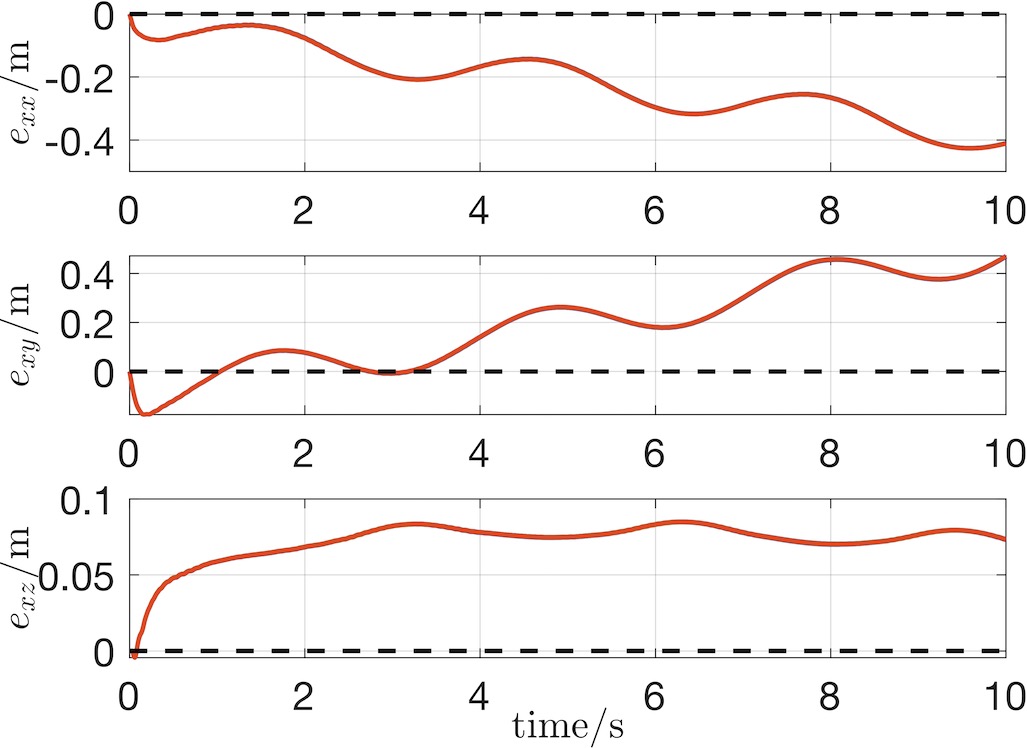}
        \vspace{-7mm}
        \caption{}
        \label{fig:position_tracking_errors_quadrotor_consensus}
    \end{subfigure}
    \hfill
    \begin{subfigure}{0.32\textwidth}
        \includegraphics[width=\linewidth]{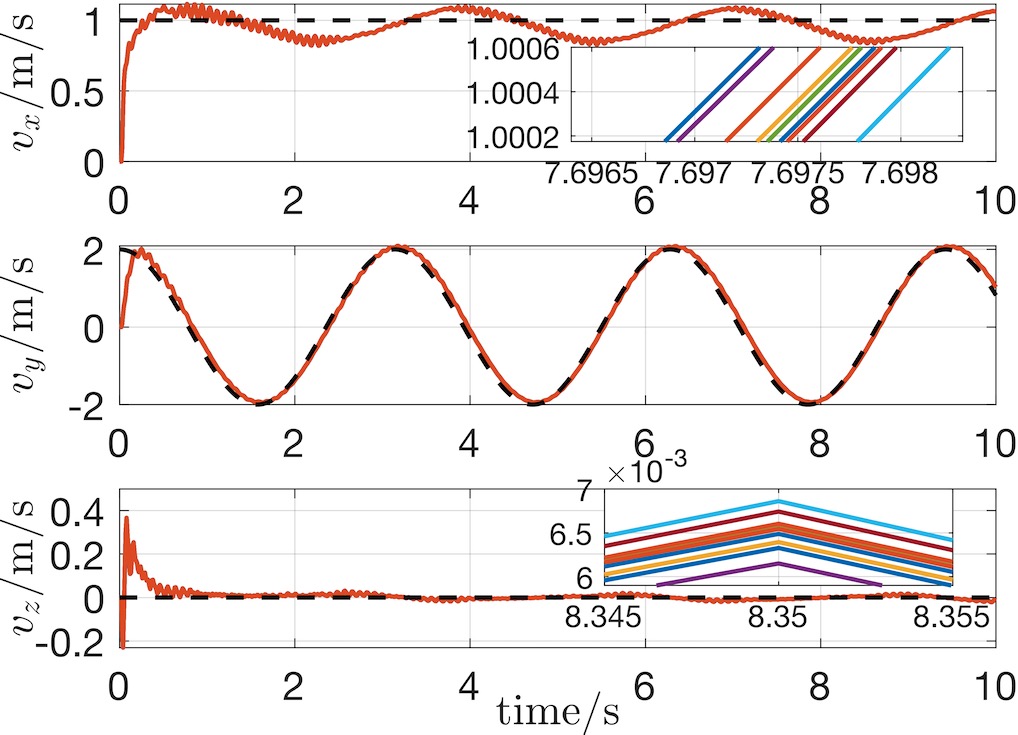}
        \vspace{-7mm}
        \caption{}
        \label{fig:velocity_tracking_results_quadrotor_consensus}
    \end{subfigure}
    \hfill
    \newline
    \begin{subfigure}{0.32\textwidth}
        \includegraphics[width=\linewidth]{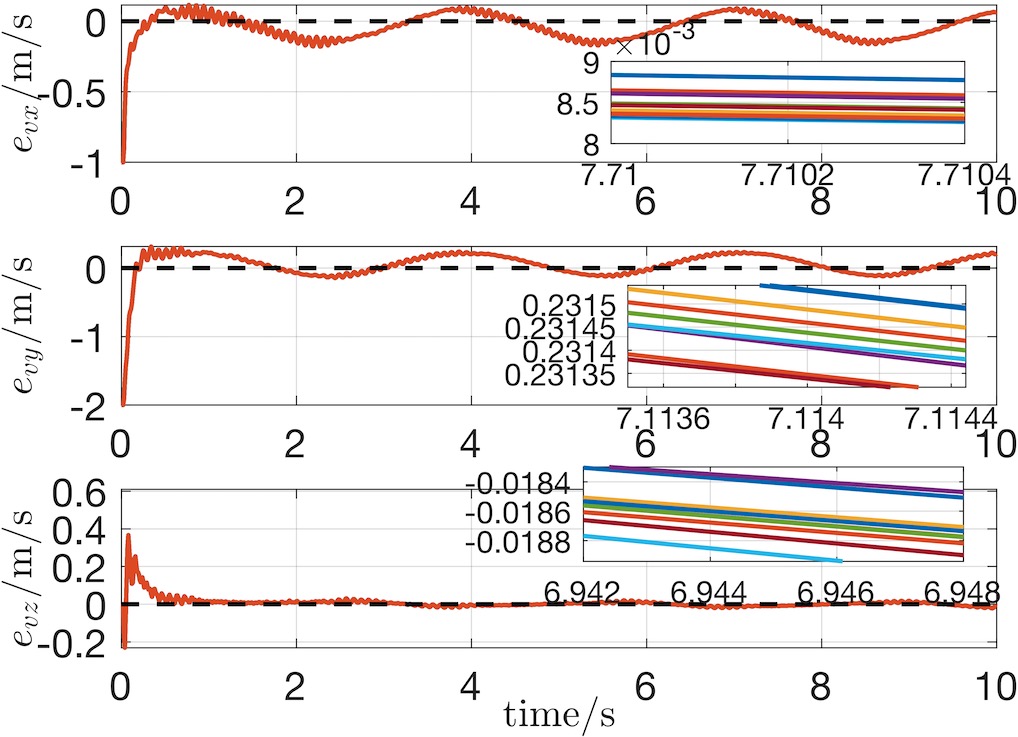}
        \vspace{-7mm}
        \caption{}
        \label{fig:velocity_tracking_errors_quadrotor_consensus}
    \end{subfigure}
    \hfill
    \begin{subfigure}{0.32\textwidth}
        \includegraphics[width=\linewidth]{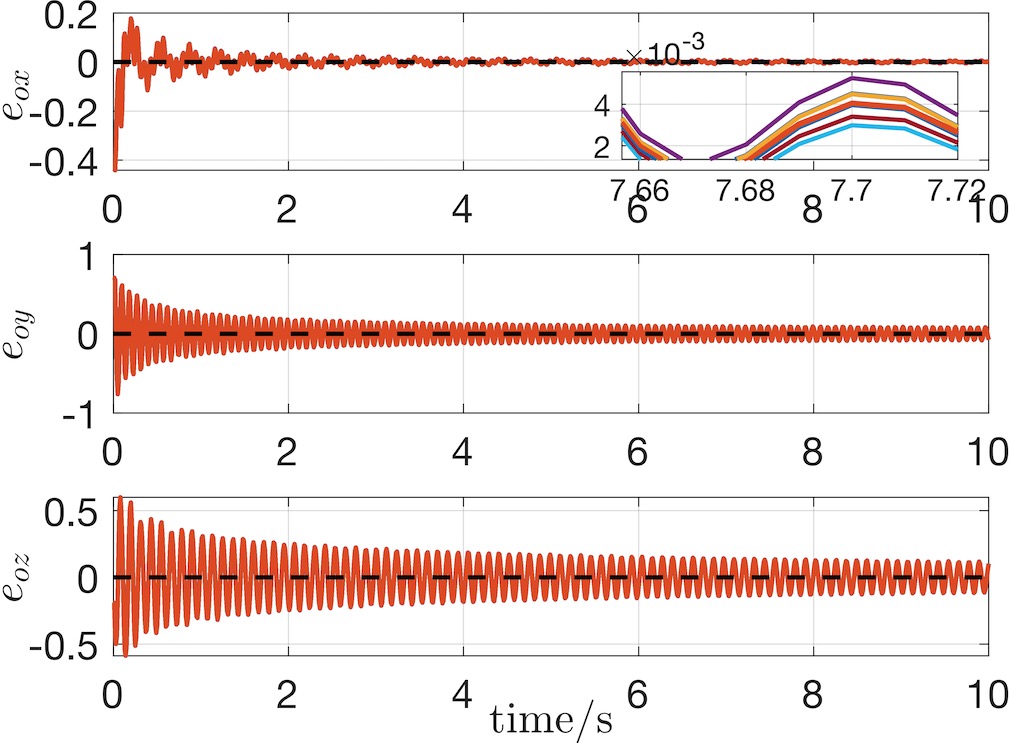}
        \vspace{-7mm}
        \caption{}
        \label{fig:orientation_tracking_errors_quadrotor_consensus}
    \end{subfigure}
    \hfill
    \begin{subfigure}{0.32\textwidth}
        \includegraphics[width=\linewidth]{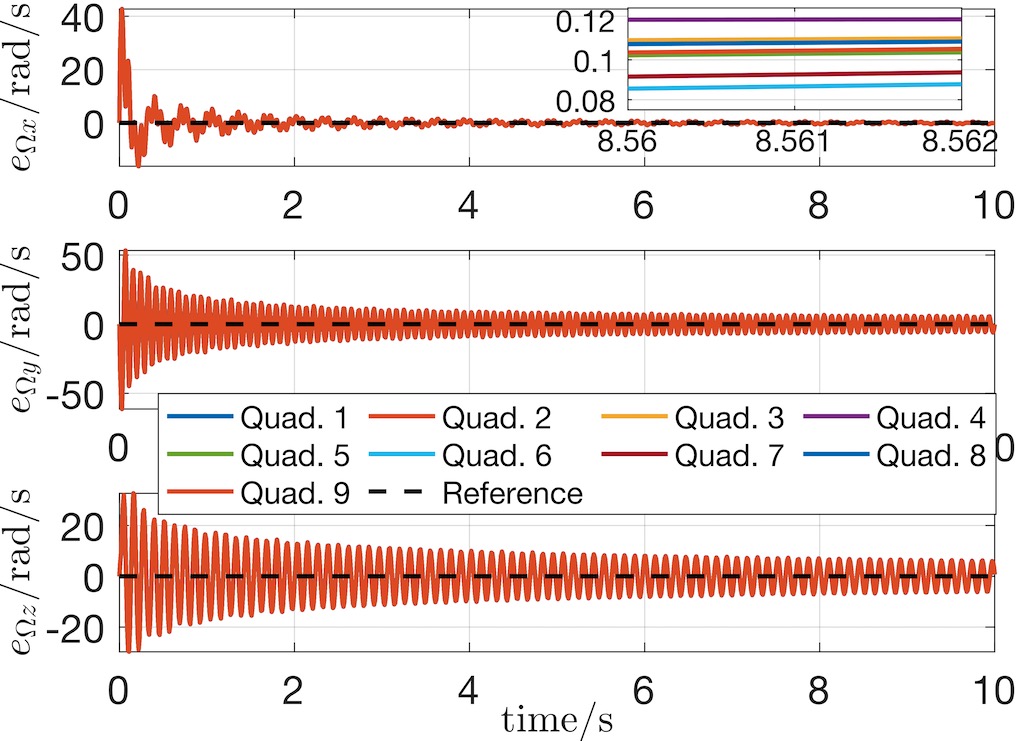}
        \vspace{-7mm}
        \caption{}
        \label{fig:angular_velocity_tracking_errors_quadrotor_consensus}
    \end{subfigure}
    \caption{Results observed by the state-of-the-art consensus-based method in \cite{amirkhani2022consensus} with $9$ followers: (a) position tracking; (b) position tracking errors; (c) translational velocity tracking; (d) translational velocity tracking errors;
    (e) orientation tracking errors; (f) angular velocity tracking errors.
     }
    \label{fig:consensus-based_controller_results}
\end{figure*}

In this section, we verify the effectiveness of our proposed decentralized co-design method in the previous section by considering a quadrotor formation control scenario ($3$ rows and $3$ columns as seen in Fig. \ref{Fig:initialTopology_quadrotor_formation}).
Simulation results are generated by a specifically developed simulator in MATLAB\footnote{Publicly available in \url{https://github.com/NDzsong2/Quadrotor-Network-Simulator.git}}.
Without loss of generality, and for ease of system and control parameters selection, we initially consider a homogeneous quadrotor formation with $8$ quadrotors following their corresponding virtual leaders. Each follower is with the parameters $m_i = 0.55\mbox{kg}$, $J_i:=\mbox{diag}\{2.2,2.9,5.3\}\times 10^{-3}\mbox{kg.m}^2$, and all these parameters are assumed to involve $\pm 10\%$ uncertainties. 
The virtual leaders of their corresponding followers are determined through their distance to a point starting at $x_0(0):=[2, -2.5, -5]^\T$, and each virtual leader's position $x_{0i}$, for all $i\in\mathcal{I}_8$, is determined by two offsets with respect to this point, i.e., row offsets $r_{i}=-(x_{mi}+x_{vi})e_1$ and column offsets $c_{i}=(y_{mi}+y_{vi})e_2$, based on their positions in the formation, where $e_1:=[1, 0, 0]^\T$, $e_2:=[0, 1, 0]^\T$. Here, $x_{mi}=2$, $x_{vi}\sim\mbox{U}(-0.5,0.5)$, and $y_{mi}=2.5$, $y_{vi}\sim\mbox{U}(-0.5,0.5)$  are the mean and variance of the row offsets and column offsets, respectively.
For example, the virtual leader at the $2\mbox{nd}$ row and $2\mbox{nd}$ column (the virtual leader of the $5\tsup{th}$ follower, i.e., $x_{05}(0)$) is selected as $x_{05}(0)=x_0(0)+2r_i+c_i$.
Besides, each follower tracks the same desired velocity and desired acceleration as $v_{0i}(t):=[1, 2\cos(2t), 0]^\T$, $\dot{v}_{0i}(t):=[0, -4\sin(2t), 0]^\T$, for all $t\geq 0$ and $i\in\mathcal{I}_8$, i.e., all the followers track their sinusoidal trajectories and keep certain gaps between each other.
The initial positions of all the followers are assumed to be $x_i(0)=x_{0i}(0)$, for all $i\in\mathcal{I}_{8}$. 
The initial translational velocity, orientation and angular velocity of all the followers are assumed as $v_i(0)=[0, 0, 0]^\T$, $R_i(0)=\I$ and $\Omega_i(0)=[0, 0, 0]^\T$, respectively.
The external disturbances are assumed as random noise, i.e., $d_{vi}$, $d_{\Omega i}\sim\mathcal{N}(0,0.01\I)$.


The initial topology of the quadrotor formation is selected as that in Fig. \ref{Fig:initialTopology_quadrotor_formation}, where each quadrotor can only communicate with its virtual leader (not shown in Fig. \ref{Fig:initialTopology_quadrotor_formation}) and its near neighbors.
To stabilize the inner-loop error dynamics of each quadrotor, i.e., the dynamics of $e_{Ri}$ and $e_{\Omega i }$ in \eqref{Eq:rigidbody_error_dynamics_original}, we select $k_{Ri} = k_{\Omega i} = 50$ in \eqref{Eq:u_2i_designed}, for all $i\in\mathcal{I}_8$. With these control parameters, the performance of the inner-loop dynamics are shown in Fig. \ref{fig:orientation_tracking_errors_quadrotor} and Fig. \ref{fig:angular_velocity_tracking_errors_quadrotor}, and it is shown that the orientation and angular velocity tracking errors converge for all directions with some small random oscillations around $0$ due to the random external disturbances. 

Besides, the inner-loop errors do not amplify over the formation, and the maximum infinite norm of the inner-loop errors is $\|[e_{ai}^\T,e_{\Omega i}^\T]^\T\|_{\infty}=\|e_{\Omega y}\|_{\infty}=0.86$, for all $i\in\mathcal{I}_7$, which ensures not only the boundedness of the inner-loop errors, but also the sMS of the inner-loop error dynamics.
The tracking results of the outer-loop dynamics, i.e., $x_i$, $v_i$, $e_{xi}$ and $e_{vi}$, are illustrated in Figs. \ref{fig:position_tracking_results_quadrotor}-\ref{fig:velocity_tracking_errors_quadrotor}. 
It is readily shown in Fig. \ref{fig:position_tracking_results_quadrotor} and Fig. \ref{fig:velocity_tracking_results_quadrotor} that both the desired position and velocity signals in all directions can be well tracked within $5$s, under our proposed decentralized strategy. Similar to the inner-loop dynamics, the tracking signals involve some random oscillations around $0$ due to the presence of random disturbances and the $L_2$-gain value is $3.21$. 
We also observe a relatively larger oscillation of the position and velocity tracking in $y$-direction as in Fig. \ref{fig:position_tracking_errors_quadrotor} and Fig. \ref{fig:velocity_tracking_errors_quadrotor} with the maximal steady-state deviations of $0.6$ and $1.2$, respectively. This is due to the fact that the quadrotors have to frequently change their heading directions in order to track the sinusoidal trajectory.
Moreover, from Fig. \ref{fig:position_tracking_errors_quadrotor} and Fig. \ref{fig:velocity_tracking_errors_quadrotor}, it is obviously observed that the position and velocity tracking errors in all directions are uniformly bounded even if there exist certain random oscillations over the entire platoon with the maximum infinite norm being $\|[e_{xi}^\T,e_{vi}^\T]^\T\|_{\infty}=\|e_{vx}\|_{\infty}=2$, for all $i\in\mathcal{I}_7$. In other words, the position and velocity tracking errors do not propagate and amplify over the formation, and thus, the sMS of the outer-loop dynamics is satisfied.
Another benefit of our proposed decentralized co-design method is that the $L_2$-gains of the network always stay at the value of $3.21$ during quadrotors merging and splitting maneuvers as in Figs. \ref{fig:decentralized_topology_sMS_8quadrotors}-\ref{fig:decentralized_topology_sMS_8quadrotors_onemissing}.
This shows the robustness of our proposed method, since the merging process does not change the $L_2$-stability behavior dramatically.

The synthesized communication topologies are shown in Fig. \ref{Fig:decentralized_sMS_quadrotors_topology} based on our proposed decentralized co-design approach. 
From the synthesized topologies in Fig. \ref{Fig:decentralized_sMS_quadrotors_topology}, we observe that the links from backwards quadrotors to front ones are more dominant as compared to the reverse ones, especially the dense connection of the $1$st and $2$nd quadrotors, which implies that the information from backwards quadrotors is more important than the front ones, and the knowledge of the leader's information is more effective for quadrotor formation control with mesh stability guarantee.
This observation is practically meaningful, since the front quadrotors will propagate the errors over the network, and thus, more information from the backward quadrotors is required to ensure the non-propagation of the errors (i.e., mesh stability) and safety.

To show the compositionality of our proposed decentralized co-design, we compare the changes of the synthesized control/topology gains in the quadrotors' merging and splitting processes as in Fig. \ref{fig:decentralized_topology_sMS_8quadrotors}-\ref{fig:decentralized_topology_sMS_8quadrotors_onemissing}.
Without loss of generality, we show the changes of the controller/topology gain $K_{13}$ from the $3$rd follower to the $1$st follower during the maneuverings. Based on the simulation results, the gain $K_{13}$ keeps to be $\scriptsize\bmtx{\0 & \0 \\ k_{13}^x & k_{13}^v}$, where 
$k_{13}^x=\scriptsize\bmtx{-0.66 & -0.66 & -0.66 \\ -0.66 & -0.66 & -0.66 \\ -0.66 & -0.66 & -0.66}*10^{-4}$ and 
$k_{13}^v=\scriptsize\bmtx{0.795 & 0.795 & 0.795 \\ 0.795 & 0.795 & 0.795 \\ 0.795 & 0.795 & 0.795}*10^{-4}$, during the merging of the $4$th-$7$th followers. 


To better illustrate the effectiveness of the proposed control method in this paper, 
we compare our proposed co-design methods to a state-of-the-art consensus-based controller \cite{amirkhani2022consensus}. For this controller, we assume that the quadrotors are interconnected with their direct neighbors, i.e., each quadrotor in the formation can only communicate with its nearest neighbors without the knowledge of other quadrotors' information. 
Simulation results of the consensus-based method are shown in Fig. \ref{fig:consensus-based_controller_results} for the quadrotor formation control with $9$ followers.
The major improvement is on the translational velocity and inner-loop tracking performance as seen in Fig. \ref{fig:velocity_tracking_errors_quadrotor_consensus}, Figs. \ref{fig:orientation_tracking_errors_quadrotor_consensus} and \ref{fig:angular_velocity_tracking_errors_quadrotor_consensus}, respectively, where the translational velocity tracking errors, i.e., $e_{vi}$, and the inner-loop tracking errors, i.e., $e_{ai}$ and $e_{\Omega i}$ converge rapidly within $0.5$s and $2.32$s, respectively. Besides, the maximal steady-state deviations are $0.23$ and $7.57$ for $e_{vi}$, and inner-loop tracking errors, respectively.


However, compared to our proposed co-design method, the consensus-based method causes frequent oscillations of the inner-loop tracking errors $e_{ai}$ and $e_{\Omega i}$ with larger magnitudes as seen in Figs. \ref{fig:orientation_tracking_errors_quadrotor_consensus} and \ref{fig:angular_velocity_tracking_errors_quadrotor_consensus}, and even the divergence of the position tracking errors on the $x$- and $y$-directions as seen in Figs. \ref{fig:position_tracking_results_quadrotor_consensus} and \ref{fig:position_tracking_errors_quadrotor_consensus}, which means that the quadrotors will gradually drift away from the desired trajectory.
This is probably because there is more information from the front quadrotors to the backward ones, which causes a more complex error propagation over the formation. Hence, the mesh stability condition is not satisfied. Moreover, the minimum $L_2$-gain is $47.5$, which is obviously larger than our proposed co-design method.
Another point to be noted is that the consensus-based controller is not compositional, since the designed controller needs to be redesigned (i.e., re-tuned) when quadrotors merge or split.
Based on the above observations, the performance of our proposed co-design method is highlighted in terms of scalability, compositionality, and tracking performance.

\section{Conclusion}\label{sec:conclusion}

In this paper, we proposed a dissipativity-based decentralized control and topology co-design framework for rigid body networks.
The proposed method can not only achieve the basic tracking control objectives, but also ensure the scalable mesh stability and the compositionality of the rigid body networks, which enables the merging and splitting control of the agents.
Besides, the synthesized topology under our proposed method indicates that the information from the backward neighbors is more dominant in network control and ensuring the mesh stability. 
Furthermore, the performance of our proposed co-design method is highlighted via a comparison study with respect to a state-of-the-art consensus-based method.
Future work aims to extend the results to the case when not all followers know the leader's information.

\bibliographystyle{IEEEtran}
\bibliography{references}

\end{document}